\documentclass[11pt,a4paper]{article}
\usepackage{fullpage}
\usepackage{microtype}
\usepackage{mathtools,amsmath} \usepackage{amssymb,amsthm,lineno}
\usepackage{graphicx} \usepackage{xspace} \usepackage{xypic}
\usepackage[numbers]{natbib} \usepackage{hyperref}
\usepackage{color,pdfpages}
\graphicspath{{figures/}}
\theoremstyle{plain} 

\newtheorem{theorem}{Theorem} 
\newtheorem{lemma}[theorem]{Lemma}

\newtheorem{corollary}[theorem]{Corollary}

\title{Approximate Range Queries for Clustering%
  \footnote{This research was supported by NRF grant 2011-0030044 (SRC-GAIA) funded by the government of Korea,
	and the MSIT (Ministry of Science and ICT), Korea, under the SW Starlab support program (IITP-2017-0-00905) 
	supervised by the IITP (Institute for Information \& communications Technology Promotion).}}

\author{Eunjin Oh\thanks{Max Planck Institute for Informatics, {Saarbr\"{u}cken, Germany, Email: {\tt{eoh@mpi-inf.mpg.de}}}} \and
  Hee-Kap Ahn\thanks{Pohang University of Science and Technology, {Pohang, Korea}, Email: {\tt{heekap@postech.ac.kr}}}}

\newcommand{\median}{\ensuremath{\Phi_\textsf{M}}}
\newcommand{\means}{\ensuremath{\Phi_\textsf{m}}}
\newcommand{\ccenter}{\ensuremath{\Phi_\textsf{c}}}

\newcommand{\eps}{\varepsilon}

\newcommand{\opt}{\ensuremath{\textsc{Opt}_k}}
\newcommand{\dist}[2]{\ensuremath{d(#1,#2)}}

\newcommand{\setdist}[2]{\ensuremath{d(#1,#2)}}
\newcommand{\oct}{\ensuremath{\mathcal{T}_c}}
\newcommand{\octs}{\ensuremath{\mathcal{T}_s}}
\newcommand{\cell}{\scalebox{0.9}{\ensuremath{\square}}}

\newcommand{\tR}{\bar{R}}
\newcommand{\ccell}{\overline{\cell}}
\newcommand{\lb}{\textsc{lb}}
\newcommand{\apx}{\textsc{apx}}
\newcommand{\pset}{\ensuremath{\mathcal{C}}}
\newcommand{\proj}[1]{\ensuremath{I(#1)}}
\newcommand{\tl}{\ensuremath{<_t}}
\newcommand{\sfloor}[1]{\ensuremath{\lfloor #1 \rfloor_s}}
\newcommand{\sceil}[1]{\ensuremath{\lceil #1 \rceil_s}}
\newcommand{\meb}{\ensuremath{\textsc{meb}}}

\begin{document}
\date{}
\maketitle

\begin{abstract}
We study the approximate range searching for three variants of the
clustering problem  
with a set $P$ of $n$ points in $d$-dimensional Euclidean space and
axis-parallel rectangular range queries: the \emph{$k$-median}, \emph{$k$-means}, and
\emph{$k$-center range-clustering query problems}. We present
data structures and query algorithms that compute
$(1+\eps)$-approximations to the optimal clusterings of $P\cap Q$ efficiently
for a query consisting of an orthogonal range $Q$, an integer $k$, and a value $\eps>0$.
\end{abstract}

\section{Introduction} 
Geometric range searching asks to preprocess a set of objects and to build a data structure so that 
all objects intersecting a given query range can be reported quickly.
There are classical variants of this problem such as computing the number of objects intersecting
a query range, checking whether an object intersects a query range, and finding the closest
pair of objects intersecting a query range.
It is a widely used technique
in computer science with numerous applications in geographic information systems,
computer vision, machine learning, and data mining.
These range searching problems have been studied extensively in computational
geometry over the last decades.
For more information on geometric range searching, refer to the surveys by
Matousek~\cite{Matousek-1994}, and Agarwal and Erickson~\cite{Agarwal-1999},
and the computational geometry book~\cite{CGbook}.

However, there are a large number of objects intersecting a query range in many real-world applications,
and thus it takes long time to report all of them. Thus one might want to obtain
a property of the set of such objects instead of obtaining all such objects.
Queries of this kind are called \emph{range-analysis queries}.
More formally, the goal of this problem is to preprocess a set $P$ of objects
with respect to a fixed range-analysis function $f$ 
and to build a data structure so that $f(P\cap Q)$ can be computed efficiently for any query 
range $Q$. 
These query problems have been studied extensively under various range-analysis functions
such as the diameter or width of a point set~\cite{Brass-2013} and the length of the minimum
spanning tree of a point set~\cite{Arya-2015}.
Note that the classical variants mentioned above are also range-analysis query problems.
A clustering cost can also be considered as a range-analysis function.

Clustering is a fundamental research topic in computer science and arises in various
applications~\cite{Jain-1999}, including pattern recognition and classification, data mining,
image analysis, and machine learning. 
In clustering, the objective is to group a set of data points into clusters so that
the points from the same cluster are similar to each other and points from
different clusters are dissimilar.
Usually, input points are in a high-dimensional space and the similarity is
defined using a distance measure. There are a number of variants of the clustering
problem in the geometric setting 
depending on the similarity measure such as 
the $k$-median, $k$-means, and $k$-center clustering problems.

In this paper, we study the approximate range-analysis query problems for three variants of the
clustering with a set $P$ of $n$ points in $d$-dimensional Euclidean space
with $d\geq 2$ and
axis-parallel rectangular range queries: the \emph{$k$-median},
\emph{$k$-means}, and \emph{$k$-center range-clustering query problems}.
The approximate $k$-median, $k$-means and $k$-center range-clustering query problems are
defined as follows: Preprocess $P$
so that given a query range $Q$, an integer $k$ with $1\leq k\leq n$ and a value $\eps>0$
as a query, an $(1+\eps)$-approximation to the $k$-median, $k$-means, and $k$-center clusterings
of $P\cap Q$ can be computed efficiently. 
Our desired query time is polynomial to $\log n$, $k$ and $(1/\eps)$. 

\subsection{Previous Work}
The $k$-median and $k$-means clustering problems have been studied extensively and there are 
algorithms achieving good approximation factors and polynomial running times to the
problem. Har-Peled and Mazumdar~\cite{Har-Peled-2004} presented an $(1+\eps)$-approximation
algorithm for the $k$-means and $k$-median clustering using
coresets for points in $d$-dimensional Euclidean space.
Their algorithm constructs a
$(k,\eps)$-coreset with property that for any arbitrary set $C$ of $k$ centers,
the clustering cost on the coreset with respect to $C$
is within $(1\pm\eps)$ times
the clustering cost on the original input points with respect to $C$. Then it
computes the clusterings for the coreset using a known weighted
clustering algorithm.
Later, a number of algorithms for computing smaller coresets for the $k$-median and $k$-means clusterings
have been presented~\cite{Chen-2009,Feldman-2011,Har-Peled2007-smaller}.
The smallest $(k,\eps)$-coresets known so far have size of $O(k/\eps^2)$ for both $k$-median and $k$-means~\cite{Feldman-2011}.

The $k$-center clustering problem has also been studied extensively. 
It is NP-Hard to approximate the 2-dimensional $k$-center problem within a factor of less than 2 even under
the $L_{\infty}$-metric~\cite{Feder}.
A $2$-approximation to the $k$-center can be computed in $O(kn)$ time for 
any metric space~\cite{Feder}, and in $O(n\log k)$ time  for any $L_p$-metric space~\cite{Gon1985}.
The exact $k$-center clustering can be computed in $n^{O(k^{1-1/d})}$ time in $d$-dimensional
space under any $L_p$-metric~\cite{agarwal2002}.
This algorithm combined with a technique for grids yields an $(1+\eps)$-approximation algorithm
for the $k$-center problem that takes $O(n\log k+ (k/\eps)^{O(k^{1-1/d})})$ time for any
$L_p$-metric~\cite{agarwal2002}.
Notice that all these algorithms are \emph{single-shot} algorithms, that is, they compute
a clustering of given points (without queries) just once.

There have been some results on range-analysis query problems related to
clustering queries. 
Brass et al.~\cite{Brass-2013} presented data structures of finding
extent measures: the width, area, or perimeter of the convex hull of
$P\cap Q$ and the smallest enclosing disk of $P\cap Q$. Arya et al.~\cite{Arya-2015} studied
data structures that support clustering queries on
the length of the minimum spanning tree of $P\cap Q$.
Various types of range-aggregate nearest neighbor queries have also been
studied~\cite{Papadias-2005,Shan2003}.

Nekrich and Smid~\cite{Nekrich-2010} considered approximate
range-aggregate queries such as the diameter or radius of the smallest
enclosing ball for points in $d$-dimensional space. Basically, their algorithm constructs
a $d$-dimensional range tree as a data structure,
in which each node stores a $\delta$-coreset of points in its subtree (but not explicitly), 
and applies range-searching query algorithms on the tree, where $\delta$ is
a positive value. Their algorithm works
for any aggregate function that can be approximated using
a decomposable coreset including coresets for the $k$-median, $k$-means and $k$-center
clusterings. 
In this case, the size of the data structure is $O(kn\log^d n/\delta^2)$, and the
query algorithm computes a $(k,\delta)$-coreset of size
$O(k\log^{d-1} n/\delta^2)$.
However, their algorithm uses $k$ and $\delta$ in constructing
the data structure for the clusterings, and therefore $k$ and $\delta$ are fixed
over range-clustering queries.

Very recently, Abrahamsen et al.~\cite{Abrahamsen-2017} considered
$k$-center range-clustering queries for $n$ points in $d$-dimensional space. 
They presented a method, for a query consisting of a range $Q$,
an integer $k$ with $1\leq k\leq n$ and a value $\eps>0$, of computing a $(k,\eps)$-coreset 
$S$ from $P\cap Q$
of size $O(k/\eps^d)$ in $O(k(\log n/\eps)^{d-1}+ k/\eps^d)$ time such that 
the $k$-center of $S$ is an $(1+\eps)$-approximation to the $k$-center of $P\cap Q$.
After computing the coreset, they compute an $(1+\eps)$-approximation to the $k$-center of the
coreset using a known single-shot algorithm. 
Their data structure is of size $O(n\log^{d-1} n)$ and its query algorithm computes 
an $(1+\eps)$-approximate to a $k$-center range-clustering query
in $O(k(\log n/\eps)^{d-1}+ T_\textnormal{ss}(k/\eps^d))$ time, where
$T_\textnormal{ss}(N)$ denotes the running time
of an $(1+\eps)$-approximation single-shot algorithm for the $k$-center  problem on $N$ points.

The problem of computing the diameter of input points contained in a query range
can be considered as a special case of the range-clustering problem. 
Gupta et al.~\cite{Gupta} considered this problem in the plane and presented
two data structures. One requires $O(n\log^2 n)$ size that supports queries
with arbitrary approximation factors $1+\eps$ in $O(\log n/\sqrt{\eps}+\log^3 n)$
query time and the other requires a smaller size $O(n\log n/\sqrt{\delta})$ that
supports only queries with the \textit{fixed} approximation factor $1+\delta$
with $0<\delta<1$  that is used for constructing the data structure.
The query time for the second data structure is  $O(\log^3 n/\sqrt{\delta})$.
Nekrich and Smid~\cite{Nekrich-2010} presented a data structure for this problem in a
higher dimensional space that has size $O(n \log^d n)$ and supports
diameter queries with the fixed approximation factor  
$1+\delta$ in $O(\log^{d-1} n/\delta^{d-1})$ query time.
Here, $\delta$  is an approximation
factor given for the construction of their data structure, and
therefore it is fixed for queries to the data structure. 

\subsection{Our Results}
We present algorithms for $k$-median, $k$-means,
and $k$-center range-clustering queries with query times polynomial
to $\log n$, $k$ and $1/\eps$. These algorithms have a similar structure:
they compute a $(k,\eps)$-coreset
of input points contained in a query range, and then compute a clustering
on the coreset using a known clustering algorithm.
We call an algorithm for computing a clustering of given points (without queries) a
\emph{single-shot algorithm} for the clustering.
We use $T_\textnormal{ss}(N)$ to denote the running time
for any $(1+\eps)$-approximation single-shot algorithm of each problem on $N$ points.
For a set $P$ of $n$ points in
$d$-dimensional Euclidean space with $d\geq 2$, we present the
following results.

\begin{itemize}
\item There is a data structure of size $O(n\log^d n)$ such that an
  $(1+\eps)$-approximation to the $k$-median or $k$-means
    clustering of $P\cap Q$ can be computed in time
  \[O(k^5\log^9 n+ k\log^d n/\eps + T_\textnormal{ss}(k\log n/\eps^d))\] for any orthogonal
  range $Q$, any integer $k$ with $1\leq k\leq n$, and any value $\eps>0$ given
  as a query. To our best knowledge, this is the first result on the $k$-median
    and $k$-means clusterings
    for orthogonal range queries with any integer $k$ and any value $\eps$.
%

\item There is a data structure of size $O(n\log^{d-1} n)$ such that
   an $(1+\eps)$-approximation to the $k$-center clustering of $P\cap Q$ can be computed in time 
  \[O(k\log^{d-1} n+k\log n/\eps^{d-1}+T_\textnormal{ss}(k/\eps^d))\] for any orthogonal range
  $Q$, an integer $k$ with $1\leq k\leq n$, and a value $\eps>0$ given as a
  query.
  This improves the result by Abrahamsen et al.~\cite{Abrahamsen-2017}.
  
\item There is a data structure of size $O(n\log^{d-1} n)$ such that
	an $(1+\eps)$-approximation to the diameter (or radius) of $P\cap Q$ can be computed in time 
	\[O(\log^{d-1} n+\log n/\eps^{d-1})\] 
	for any orthogonal range $Q$ and a value $\eps>0$ given as a query. 
    This improves the results by Nekrich and Smid~\cite{Nekrich-2010}.
\end{itemize}

Our results are obtained by combining range searching with coresets.
The $k$-median and $k$-means range-clusterings have not been studied before, except
the work by Nekrich and Smid. They presented a general method to compute 
approximate range-aggregate queries. Their approach can be used
to compute an $(1+\delta)$-approximation to the $k$-median or $k$-means
range-clustering for a positive value $\delta$
which is given in the construction of their data structure. 
However, it is not clear how to use or make their data structure
  to support approximate range-clustering queries with various approximation
  factors
we consider in this paper unless those values are known in advance.
Indeed, the full version of the paper by Abrahamsen et al.~\cite{Abrahamsen-2017}
  poses as an open problem a data structure supporting $(1+\eps)$-approximate
$k$-median or $k$-means range-clustering
queries with various values $\eps$.
Our first result answers to the question and presents
a data structure for the $k$-median and $k$-means range-clustering problems
for any value $\eps$.

Our second result, that is, the data structure and its query algorithm for the
$k$-center problem, improves the best known previous work by Abrahamsen et
al.~\cite{Abrahamsen-2017}.
Recall that the query algorithm by Abrahamsen et al. takes $O(k\log^{d-1} n/\eps^{d-1} + T_\textnormal{ss}(k/\eps^d))$ time. We improved the first term of their running time by
a factor of $\min\{1/\eps^{d-1},\log^{d-2} n\}$.

Our third result, that is, the data structure and its query algorithm
for computing an approximate diameter and radius of points in a query
range, improves the best known previous work by
Nekrich and Smid~\cite{Nekrich-2010}.
Our third result not only allows
queries to have arbitrary approximation factor values $1+\eps$, but
also improves the size and the query time of these data structures.
The size is improved by a factor of $\log n$.
Even when $\eps$ is fixed to $\delta$, the query time is improved
  by a factor of $\min\{1/\delta^{d-1}, \log^{d-2} n\}$ compared to the one
  by Nekrich and Smid~\cite{Nekrich-2010}.

\medskip
A main tool achieving the three results is a new data
structure for range-emptiness and  range-counting queries. 
Consider a grid $\Gamma$ with side length $\gamma$ covering an axis-parallel
hypercube with side length $\ell$
that is aligned with the standard quadtree.
For a given query range $Q$ and every cell $\cell$ of $\Gamma$,
we want to check whether there is a point of $P$ contained in $\cell\cap Q$ efficiently.
(Or, we want to compute the number of points of $P$ contained in $\cell\cap Q$.) 
For this purpose, one can use a
data structure for an orthogonal range-emptiness queries supporting
$O(\log^{d-1} n)$ query time~\cite{CGbook}.
Thus, the task takes $O((\ell/\gamma)^d \log^{d-1} n)$ time for all cells of $\Gamma$ in total.
Notice that $(\ell/\gamma)^d$ is the number of grid cells of $\Gamma$. 

To improve the running time for the task,
we present a new data structure that supports a range-emptiness query in 
$O(\log^{d-t-1} n+\log n)$ time for a cell of $\Gamma$
intersecting no face of $Q$ with dimension smaller than $t$ for any fixed $t$.
Using our data structure, the running time for the task is improved to
$O(\log^{d-1} n+ (\ell/\gamma)^d \log n)$.
To obtain this data structure, we observe that a range-emptiness query
for $\cell\cap Q$ can be reduced to a $(d-t-1)$-dimensional orthogonal
range-emptiness query on points contained in $\cell$ if $\cell$ 
intersects no face of $Q$
with dimension smaller than $t$. We maintain
a data structure for $(d-t-1)$-dimensional orthogonal range-emptiness
queries for each cell of the compressed quadtree for every $t$. However, this requires
 $\Omega(n^2)$ space in total if we maintain these data structures explicitly.
We can reduce the space complexity using a
method for making a data structure partially
persistent~\cite{Driscoll-1989}.


Another tool to achieve an efficient query time is a \emph{unified grid} based
on quadtrees.  For the $k$-median and $k$-means clusterings,
we mainly follow an approach given by Har-Peled and Mazumdar~\cite{Har-Peled-2004}.
They partition input points with respect to the approximate centers explicitly,
and construct a grid for each subset of the partition.
Then they snap each input point $p$ to a cell of the grid
constructed for the subset where $p$ is involved. 
However, their algorithm is a single-shot algorithm and requires
$\Omega(|P\cap Q|)$ time due to the computation of a coreset from
approximate centers of the points contained in a given query box.
In our algorithm, we do not partition input points explicitly but we
use only one grid instead, which we call the unified grid, for the purpose
in the implementation so that a coreset can be constructed more efficiently.

The tools we propose in this paper to implement the algorithm by
Har-Peled and Mazumdar can be used for
implementing other algorithms based on grids.  For example, if Monte
Carlo algorithms are allowed, the approach given
by Chen~\cite{Chen-2009} for approximate range queries can be
implemented by using the tools we propose in this paper.

\section{Preliminaries}
Let $P$ be a set of $n$ points in $d$-dimensional Euclidean space. For any two points
$x$ and $y$ in $d$-dimensional space, we use $\dist{x}{y}$ to denote the
Euclidean distance between $x$ and $y$.  For a point
$x$ and a set $Y$ in $d$-dimensional space, we use
$\setdist{x}{Y}$ to denote the smallest Euclidean distance between $x$
and any point in $Y$.
Throughout the paper, we use the term \emph{square} in a generic way to refer a
$d$-dimensional axis-parallel hypercube. Similarly, we use the term \emph{rectangle}
to refer a $d$-dimensional axis-parallel box.

\subsection{Clusterings}
For any integer $k$ with $1\leq k\leq n$, let $\pset_k$ be the family of the sets of 
at most $k$ points in $d$-dimensional Euclidean space.  Let
$\Phi : \pset_n\times \pset_k \rightarrow \mathbb{R}_{\geq
  0}$ be a cost function which will be defined later.  For a set $P$
of $n$ points in $d$-dimensional Euclidean space, we define $\opt(P)$ as the minimum
value $\Phi(P,C)$ over all sets $C\in\pset_k$.  We call a set
$C\in\pset_k$ realizing $\opt(P)$ a \emph{$k$-clustering of $P$ under the cost
  function $\Phi$}.

In this paper, we consider three cost functions $\median, \means$ and $\ccenter$
that define the $k$-median, $k$-means and $k$-center clusterings, respectively.
Let $\phi(p,C)=\min_{c\in C} \dist{p}{c}$ for any point $p$ in $P$.
The cost functions are defined as follows: for any set $C$ of $\pset_k$, 
\[\median(P,C)=\sum_{p\in P} \phi(p,C),\quad
\means(P,C)=\sum_{p\in P} (\phi(p,C))^2,\quad
\ccenter(P,C)=\max_{p\in P} \phi(p,C).\]

We consider the query variants of these problems. We preprocess
$P$ so that given a query rectangle $Q$, an integer $k$ with $1\leq k\leq n$, and 
a value $\eps>0$, an $(1+\eps)$-approximate
$k$-clustering of the points contained in $Q$ can be reported efficiently.  Specifically,
we want to report a set $C\in\pset_k$ with
$\Phi(P_Q,C) \leq (1+\eps) \opt(P_Q)$ in sublinear time, where
$P_Q=P\cap Q$.  We call a query of this type a \emph{range-clustering
  query}.

\subsection{Coreset for Clustering}
Consider a cost function $\Phi$. We call a set
$S\subseteq\mathbb{R}^d$ a \emph{$(k,\eps)$-coreset} of $P$ for the
$k$-clustering under the cost function $\Phi$ if the following holds:
for any set $C$ in $\pset_k$,
\[(1-\eps)\Phi(P,C) \leq \Phi(S,C) \leq (1+\eps)\Phi(P,C).\] Here, the
points in a coreset might be weighted. In this case, the distance
between a point $p$ in $d$-dimensional space and a weighted point $s$ in a
coreset is defined as $w(s)\cdot d(p,s)$, where $w(s)$ is the weight
of $s$ and $d(p,s)$ is the Euclidean distance between $p$ and $s$.

By definition, an $(1+\eps)$-approximation to the $k$-clustering of $S$
is also an $(1+\eps)$-approximation to the $k$-clustering of $P$. Thus,
$(k,\eps)$-coresets can be used to obtain a fast approximation
algorithm for the $k$-median, $k$-means, and $k$-center clusterings. A
$(k,\eps)$-coreset of smaller size gives a faster approximation
algorithm for the clusterings. The followings are the sizes of
the smallest $(k,\eps)$-coresets known so far: $O(k/\eps^2)$ 
for the $d$-dimensional Euclidean $k$-median and $k$-means clusterings~\cite{Feldman-2011}, and
$O(k/\eps^{d})$ for the $d$-dimensional Euclidean $k$-center clustering~\cite{agarwal-2005}.

It is also known that $(k,\eps)$-coresets for the $k$-median, $k$-means,
and $k$-center clusterings are \emph{decomposable}. That is, if $S_1$
and $S_2$ are $(k,\eps)$-coresets for disjoint sets $P_1$ and $P_2$,
respectively, then $S_1\cup S_2$ is a $(k,\eps)$-coreset for
$P_1\cup P_2$ by~\cite[Observation 7.1]{Har-Peled-2004}.  Using this property, one can obtain data structures on
$P$ that support an
$(1+\delta)$-approximation to the $k$-median, $k$-means, and $k$-center
range-clustering queries for constants $\delta>0$ and $k$
with $1\leq k\leq n$ which are given in the construction phase 
as follows. 

Consider the $d$-dimensional range tree on $P$, a multi-level binary search
tree~\cite{CGbook}.  There are $O(n\log^{d-1} n)$ nodes in the
  level-$d$ trees of the range tree in total.  Each such node $v$ 
corresponds to a $d$-dimensional axis-parallel box $B(v)$. For each node $v$, assume that
a $(k,\delta)$-coreset of the points of $P$ contained in $B(v)$
is stored. For any rectangle $Q$, there are $O(\log^d n)$ nodes $v$ such that
$P\cap Q$ is the set of the input points contained in the union of $B(v)$'s.
Such nodes are called \emph{canonical nodes} for $Q$. To answer a clustering query with
a query rectangle $Q$, it suffices
to return the union of the $(k,\delta)$-coresets stored in every canonical
node for $Q$, which is a $(k,\delta)$-coreset of $P\cap Q$. 
Then the query time and the size of the coreset are $O(f(k,\delta)\log^d n)$, where $f(k,\delta)$ is the 
size of a $(k,\delta)$-coreset of a clustering obtained from a single-shot
algorithm for constants $\delta>0$ and $k$ with $1\leq k\leq n$.
For the size of the data structure, observe that the size of
the coreset stored in each node $v$ is at most the
number of the points contained in $B(v)$.  The total number of
points contained in $B(v)$ for every node $v$ of the range
tree is $O(n\log^d n)$, and thus the data structure has size $O(n\log^d n)$.

One drawback with the data structure is that $k$ and $\delta$ are determined 
in the construction phase of the structure, and therefore they are fixed
over range-clustering queries. To resolve this, we construct
a number of the data structures for different values of $k=1,2,2^2,
\ldots,2^{\lceil \log n\rceil}$. Given a value $k$ as a query with
$\bar{k}\leq k< 2\bar{k}$, we simple return a
$(\bar{k},\delta)$-coreset, which is a $(k,\delta)$-coreset.
This does not increase the size of the data
structure and the query time asymptotically and allows $k$ to be a part of
queries.

\begin{lemma}\label{lem:constant-coreset}
  Given a set $P$ of $n$ points in $d$-dimensional space and a
  value $\delta>0$ given in the construction phase, we can construct a data structure of size
  $O(n\log^d n)$ so that a $(k,\delta)$-coreset of
  $P\cap Q$ for the $k$-median and $k$-means clusterings of size
  $O(k\log^{d} n)$ can be computed in $O(k\log^{d} n)$ time for
  any orthogonal range $Q$ and any integer $k$ with $1\leq k\leq n$ given as a query.
\end{lemma}

Note that the approximation factor of the coreset is still fixed to queries.
In Section~\ref{sec:median}, we will describe a data structure
and its corresponding query algorithm answering $k$-median and
$k$-means range-clustering queries that allow queries to have arbitrary
approximation factor values $1+\eps$. The query algorithm in Section~\ref{sec:median}
uses the algorithm in Lemma~\ref{lem:constant-coreset} as a subprocedure.
%

\subsection{Single-Shot Algorithms for the
  \texorpdfstring{$k$}{k}-Median and \texorpdfstring{$k$}{k}-Means
  Clusterings}
\label{sec:single-shot}
The single-shot version of this problem was studied by Har-Peled and
Mazumdar~\cite{Har-Peled-2004}.  They gave algorithms to compute 
$(k,\eps)$-coresets of size $O(k\log n/\eps^d)$ for the $k$-median and
$k$-means clusterings. Since we extensively use their results, we give
an overview to their algorithm for the $k$-median
clustering. The algorithm for the $k$-means clustering works
similarly.
In this subsection, we use $\Phi$ to denote $\median$ for ease of
description.

Their algorithm starts with computing
a constant-factor approximation $A\subset\mathbb{R}^d$
to the $k$-means clustering of $P$, that is, $A$ satisfying
$\Phi(P,A)\leq c_1 \cdot\opt(P)$ for
some constant $c_1>1$.
The approximation set consists of $O(k\log^3 n)$ centers.
Then given the constant-factor approximation set of $P$, 
it computes a $(k,\eps)$-coreset $S$ of size $O(k\log^4 n/\eps^d)$ for $P$.
From $S$, the algorithm finally obtains a smaller $(k,\eps)$-coreset of size
$O(k\log n/\eps^d)$ for $P$.

\subsubsection{Coreset from Constant-Factor Approximate Centers}
Given a constant-factor approximation $A=\{a_1,\ldots,a_m\}$
to the $k$-means clustering of $P$ such that $m$ is possibly larger than $k$, 
the algorithm by Har-Peled and Mazumdar computes
a $(k,\eps)$-coreset of size $O(|A|\log n/\eps^d)$ for $P$ as follows.
The procedure partitions $P$ with respect to $A$ into pairwise disjoint
sets $P_i$ for $i=1,\ldots,m$ such that $P_i$ consists of points $p$
in $P$ with $\dist{p}{a_i}\leq c_2\cdot\dist{p}{a_j}$ for every index
$j\neq i$ for a constant $c_2>1$. 
Note that $P_i$ is not necessarily unique.

Then it constructs a grid for each set $P_i$ with respect to $a_i$ and
snaps the points in $P_i$ to the grid as follows.  Let
$R=\Phi(P,A)/(c_1n)$, where $c_1>1$ is the approximation factor of $A$.  Let
$Q_{ij}$ be the square 
with side length $R2^j$ centered at
$a_i$ for $j=0,\ldots, M$, where $M=\lceil2\log (c_1n) \rceil$.
Let $V_{i0}=Q_{i0}$ and
$V_{ij}=Q_{ij}\setminus Q_{i(j-1)}$. To compute a grid for $P_i$,  the procedure
partitions each $V_{ij}$ into squares with side length
$r_j=\eps R2^j/(10c_1d)$. 
Figure~\ref{fig:grid}(a) illustrates
a grid constructed with respect to an approximate center in the middle.

\begin{figure}
  \begin{center}
    \includegraphics[width=0.8\textwidth]{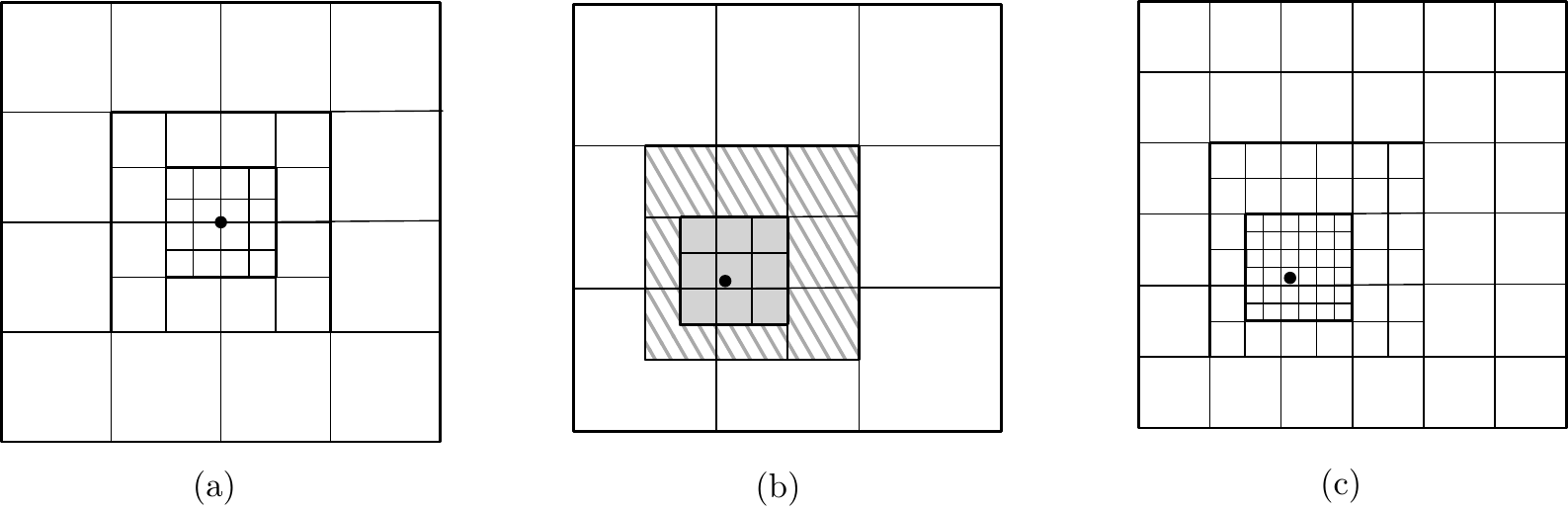}
    \caption{\small (a) Exponential grid for the single-shot algorithm
      constructed with respect to the point in the middle.  (b)
      Exponential grid for an algorithm for the query version
      constructed with respect to the point, say $a_1$, in the middle. The point
      is not the center of the grid. The gray region is the grid cluster for $Q_{10}$, 
      the dashed region is the grid cluster for $Q_{11}$, and the largest box
      is the grid cluster for $Q_{12}$.
      (c) Only first-level grid is depicted.\label{fig:grid}}
  \end{center}
\end{figure}

For every grid cell $\Box$ containing a point of $P_i$, the procedure picks
an arbitrary point $q$ of
$P_i$ contained in it and assigns the number of points of $P_i$ contained in
$\Box$ to $q$ as its weight.  Let $S_i$ be the set of
all such weighted points of $P_i$ for $i=1,\ldots, m$.
They showed that the union of all $S_i$ is a $(k,\eps)$-coreset for $P$ of size
$O(|A|\log n/\eps^d)$.

\begin{lemma}[\cite{Har-Peled-2004}]\label{lem:approx-to-coreset}
  Given a constant-factor approximation $A$ to the $k$-means clustering of $P$
  consisting of possibly more than $k$ centers, a
  $(k,\eps)$-coreset of $P$ for the $k$-median clustering of size
  $O(|A|\log n/\eps^d)$ can be computed in $O(n \log |A|)$ time.
\end{lemma}

\subsubsection{Smaller Coreset}
  By Lemma~\ref{lem:approx-to-coreset}, the algorithm constructs
  a $(k,\eps)$-coreset $S$ of size $O(k\log^4 n/\eps^d)$ using 
  a constant-factor approximation of size $O(k\log^3 n)$.  Using the coreset $S$,
  the algorithm obtains a smaller coreset of size $O(k\log n/\eps^d)$ as
  follows.
The algorithm computes a constant-factor approximation $\pset_0$ to the
$k$-center clustering of $S$ using the algorithm
in~\cite{Gon1985}. 
This clustering is an $O(n)$-approximation to the
$k$-median clustering.  Then it applies the local
search algorithm by Arya et al.~\cite{Arya-2004} 
to $\pset_0$ and $S$ to obtain a constant-factor approximation of
$P$ of size at most $k$. It
uses this set to compute a $(k,\eps)$-coreset of size
$O(k\log n/\eps^d)$ by applying Lemma~\ref{lem:approx-to-coreset} again.

\section{Data Structures for Range-Clustering Queries}\label{sec:DS}
We maintain two data structures constructed on $P$. One is
a compressed quadtree~\cite{Aluru-2005}, and the other is a variant of a range tree,
which we introduce in this paper.

\subsection{Compressed Quadtree}\label{sec:quad}
We use the term \emph{quadtrees} in a generic way to refer the
hierarchical spatial tree data structures for $d$-dimensional data
that are based on the principle of recursive decomposition of space,
also known as quadtrees, octrees, and hyperoctrees for spatial data
in $d=2, 3,$ and higher dimensions, respectively.
A \emph{standard quadtree} on $P$ is a
tree each of whose nodes $v$ corresponds to a square cell.  
The root of the quadtree
corresponds to the axis-parallel square containing all points of $P$.  A node $v$ of the quadtree corresponding
to cell $\cell$ has $2^d$ child nodes that correspond to the $2^d$
equal sized squares formed by splitting $\cell$ by $d$ axis-parallel
cuts through the center of $\cell$ if $\cell$ contains at least two points of $P$. Otherwise, $v$
is a leaf node of the quadtree.

Without loss of generality, we assume that
the side length of the square corresponding to the root is $1$. 
Then every cell of the standard quadtree has side length
of $2^{-i}$ for an integer $i$ with $0\leq i\leq t$ for some constant $t$.
 We call a value of form $2^{-i}$ for an integer $i$ with $0\leq i\leq t$ 
 a \emph{standard length}.
Also, we call a grid a \emph{standard grid}
if every cell of the grid is also in the standard quadtree. In other words, any grid
aligned with the standard quadtree is called a standard grid.

A \emph{compressed quadtree} on $P$ is a tree obtained by contracting
the edges incident to each node having only one child node in the standard
quadtree on $P$. It has size $O(n\log n)$ and can be constructed in $O(n\log n)$
time for any fixed dimension~\cite{har2011geometric}.
By definition, for every node $v$ in the compressed quadtree,
there is a node $w$ in the standard quadatree whose cell coincides with the cell of $v$.
We use $\octs$ and $\oct$ to denote the standard and compressed quadtrees constructed on
$P$, respectively.

For ease of description, we will first introduce our algorithm in
terms of the standard quadtree.  But the algorithm will be implemented
using the compressed quadtree to reduce the space complexity.
 To do
this, we need the following lemma. In the following,
  we use a node and its corresponding
cell of $\octs$ (and $\oct$) interchangeably.
For a cell $\cell$ of the standard
quadtree on $P$, there is a unique cell $\ccell$ of the compressed
quadtree on $P$ satisfying $\cell \cap P=\ccell\cap P$. We call this cell
the \emph{compressed cell} of $\cell$. 

\begin{lemma}[\cite{har2011geometric}]\label{lem:cell-query}
	Given a cell $\cell$ of $\octs$, we can find
	the compressed cell of $\cell$ in $O(\log n)$ time.
\end{lemma}

We store the points of $P$ in an array of length $n$ in a specific
order, which is called the \emph{$\mathcal{Z}$-order} defined as
follows.  Consider a DFS traversal of $\oct$ that visits
the child nodes of each node in the same relative order. The
order of the nodes of $\oct$ in which the DFS visits is called the
$\mathcal{Z}$-order~\cite{har2011geometric}.  By definition, for any
cell $\cell$ of $\oct$, the points of $P$ contained in $\cell$ appear
consecutively in the array.

\subsection{Data Structure for Range-Emptiness Queries}\label{sec:emptiness-DS}
In our query algorithm, 
we consider a standard grid $\Gamma$ of side length $\gamma$ covering an axis-parallel hypercube of side length $\ell$.
For a given query range $Q$ and every cell $\cell$ of $\Gamma$,
we want to check whether there is a point of $P$ contained in $\cell\cap Q$ efficiently.
For this purpose, one can use a
data structure for orthogonal range-emptiness queries supporting
$O(\log^{d-1} n)$ query time~\cite{CGbook}.
Thus, the task takes $O((\ell/\gamma)^d \log^{d-1} n)$ time for all cells of $\Gamma$ in total.
Notice that $(\ell/\gamma)^d$ is the number of grid cells of $\Gamma$. 

However, we can accomplish this task more efficiently using the data
structure which we will introduce in this section.
Let $t$ be an integer with $0<t\leq d$.
 We use
$\tl$-face to denote a face with dimension smaller than $t$ among faces of a $d$-dimensional rectangle. 
Note that a $\tl$-face of a $d$-dimensional rectangle is its vertex if $t=1$.
Our data structure allows us to check whether a point of $P$ is contained in
$Q\cap \ccell$ in $O(\log^{d-t-1}n +\log n)$ time for a cell $\ccell$ of
$\oct$ intersecting no $\tl$-face of $Q$ with $0< t< d$. 
Here, we first compute the compressed cell $\ccell$ of each cell $\cell$ of $\Gamma$, and then
apply the query algorithm to $\ccell$.
Recall that $\cell\cap P$ coincides with $\ccell\cap P$ for any cell $\cell$ of $\Gamma$ and 
its compressed cell $\ccell$.
In this way, we can complete the task in
$O(\sum_{t=1}^{d-1} x_t\log^{d-t-1} n+ |\Gamma|\log n+\log^{d-1} n)$ time in total, where
$x_t$ is the number of cells of $\Gamma$ intersecting no $\tl$-face of $Q$ 
but intersecting
a $t$-dimensional face of $Q$. Notice that for any
cell $\cell$ intersecting $Q$, there is an integer $t$ such that $\cell$ intersects no
$\tl$-face of $Q$, but intersects 
a $t$-dimensional face of $Q$ unless
$\cell$ contains a corner of $Q$.
Here, $x_t$ is $O((\ell/\gamma)^t)$. Therefore, we can accomplish the task
for every cell of $\Gamma$ in $O(\log^{d-1} n +(\ell/\gamma)^d \log n)$ time in total.

For a nonempty subset $I$ of $\{1,\ldots,d\}$, the \emph{$I$-projection range
  tree} on a point set $A\subseteq\mathbb{R}^d$ is the range tree
supporting fractional cascading that is constructed on the projections
of the points of $A$ onto a $(d-t)$-dimensional hyperplane orthogonal
to the $i$th axes for all $i\in I$, where $t$ is the cardinality of $I$.

\begin{lemma}\label{lem:counting-query}
  Given the $I$-projection range tree on $P\cap \cell$ for every cell
  $\cell$ of $\oct$ and every nonempty subset $I$ of $\{1,\ldots,d\}$, we can
  check whether a point of $P$ is contained in $Q\cap\cell$ for
  any query rectangle $Q$ and any cell $\cell$ of $\oct$ intersecting
  no $\tl$-face of $Q$ in $O(\log^{d-t-1} n+\log n)$ time.
\end{lemma}
\begin{proof}
  Consider a subset $I$ of $\{1,\ldots,d\}$.
  We call a facet of an axis-parallel box an \emph{$I$-facet} if it is orthogonal
  to the $i$th axis for an index $i\in I$. Note that there are exactly $2|I|$ $I$-facets of $Q$.
  For a cell $\cell$ intersecting no $\tl$-face of $Q$, we claim that there is a subset $I$ of $\{1,\ldots,d\}$ of size $t$ such that no $I$-facet of $Q$ intersects 
  $\cell$. Otherwise, there is a set $I'$ of $d-t+1$ indices such that a facet orthogonal to the
  $i'$th axis intersects $\cell$ for every $i'\in I'$. The common intersection of all such facets
is a $(t-1)$-dimensional face of $Q$, and it intersects $\cell$ since both $\cell$ and $Q$ are $d$-dimensional rectangles. This contradicts the fact that $\cell$ intersects no $\tl$-face of $Q$.
  Thus, we do not need to consider the $i$th coordinates of the
  points in $\cell$ for all $i\in I$ in testing if a point of $P$ is contained
  in $Q\cap \cell$.
  
  For a set $A$ of points in $d$-dimensional space, we use $\proj{A}$ to denote the 
  projection of $A$ onto a $(d-t)$-dimensional hyperplane orthogonal to the $i$th
  axes for all $i\in I$.
  A point of $P\cap \cell$ is contained in $Q$ if and only if 
  a point of $\proj{P\cap \cell}$ is contained in $\proj{Q}$.
  By definition, the $I$-projection range tree on $P\cap \cell$ is the
  $(d-t)$-dimensional range tree on $\proj{P\cap\cell}$.
Therefore, we can check whether a point of $\proj{P\cap\cell}$
 is contained in $\proj{Q}$ in
  $O(\log^{d-t-1} n)$ time for $t<d-1$ and in $O(\log n)$ time for $t\geq d-1$. 
\end{proof}

However, the $I$-projection range trees require $\Omega(n^2)$ space in
total if we store them explicitly.  To reduce the space complexity, we use a method of  making a
data structure \emph{partially persistent}~\cite{Driscoll-1989}.
A partially persistent data structure allows us to access any elements of
an old version of the data structure
by keeping the changes on the data structure.  Driscoll et
al.~\cite{Driscoll-1989} presented a general method of making a data
structure based on pointers partially persistent. In their method, both
time and space overheads for an update are $O(1)$ amortized, and the access time
for any version is $O(\log n)$.

\subsubsection{Construction of the \texorpdfstring{$I$}{I}-Projection Range Trees}
Consider a fixed subset $I$ of $\{1,\ldots, d\}$.
We construct the $I$-projection range trees for the cells of $\oct$ 
in a bottom-up fashion, from leaf cells to the root cell.
Note that each leaf cell of the compressed quadtree contains at most
one point of $P$. We initially construct the $I$-projection range tree for
each leaf cell of $\oct$ in total $O(n)$ time. 

Assume that we already have the $I$-projection range tree for every
child node of an internal node $v$ with cell $\cell$ of $\oct$.  Note that an internal node
of the compressed quadtree has at least two child nodes and up
to $2^d$ child nodes.
We are going to construct the $I$-projection
range tree for $v$ from the $I$-projection range trees of
the child nodes of $v$.
One may consider to merge the $I$-projection range trees for
the child nodes of $v$ into one, but we do not know any efficient
way of doing it.
Instead, we construct the $I$-projection range tree for $v$ as follows.
Let $u$ be a child node of $v$ on $\oct$ that contains the largest
number of points of $P$ in its corresponding cell among all child nodes of $v$.
We insert all points of $P$ contained in the cells for the child nodes of $v$
other than $u$ to the $I$-projection range tree of $u$ to form
the $I$-projection range tree of $v$.
Here, we do not destroy the old version of the $I$-projection range tree
of $u$ by using the method by Driscoll et al.~\cite{Driscoll-1989}. Therefore, we can
still access the $I$-projection range tree of $u$.  For the
insertion, we use an algorithm that allows us to apply fractional
cascading on the range tree under insertions of points~\cite{Mehlhorn-1990}.
We do this for every subset $I$ of $\{1,\ldots, d\}$ and construct
the $I$-projection range trees for nodes of $\oct$.

In this way, we can access any $I$-projection range tree in
$O(\log n)$ time, and therefore we can check if a point of $P$ is
contained in $Q\cap \cell$ in $O(\log^{d-t-1} n+\log n)$ time for any
query rectangle $Q$ and any cell $\cell$ of $\oct$ intersecting no
$\tl$-face of $Q$ for any integer $t$ with $0< t\leq d$ by
Lemma~\ref{lem:counting-query}.

\subsubsection{Analysis of the Construction}
The construction of the dynamic range tree~\cite[Theorem
8]{Mehlhorn-1990} requires $O(\delta\log^{d-t-1} \delta)$ space on the insertions 
of $\delta$ points in $\mathbb{R}^{d-t}$. The method by Driscoll et
al. requires only $O(1)$ overhead for each insertion on the space complexity.  Thus, the
space complexity of the $I$-projection range trees for a fixed subset $I$ consisting
of $t$ indices over all cells of $\oct$ is $O(n+ \delta\log^{d-t-1} \delta)$,
where $\delta$ is the number of the insertions performed during the
construction in total.

The update procedure for the dynamic range tree~\cite[Theorem
8]{Mehlhorn-1990} takes $O(\log^{d-t-1} n)$ time if only insertions are
allowed. The method by Driscoll requires only $O(1)$
overhead for each insertion on the update time.  Thus, the construction time is
$O(n+\delta \log^{d-t-1} n)$, where $\delta$ is the number of the
insertions performed during the construction in total.

The following lemma shows that $\delta$ is $O(n\log n)$, and thus
our construction time is
$O(n\log^{d-t} n)$ and the space complexity of the data structure is
$O(n\log^{d-t} n)$ for each integer $t$ with $0<t\leq d$. Note that there are $2^d=O(1)$
subsets of $\{1,\ldots,d\}$.
Therefore, the total space complexity and construction time are $O(n\log^{d-1} n)$.

\begin{lemma}
	For a fixed subset $I$ of $\{1,\ldots,d\}$, the total number of insertions performed during the construction of all $I$-projection range trees for every node of $\oct$ is $O(n\log n)$.
\end{lemma}
\begin{proof}
	We consider a fixed subset of $\{1,2,\ldots,d\}$, and compute the 
	number of insertions performed during the construction of the $I$-projection range trees
	for all cells of $\oct$.
	We use a notion, the
	\emph{rank} of a cell of $\oct$, to analyze the number of
	insertions performed during the construction.
	Each leaf cell of $\oct$ has rank $0$. For an internal node with cell $\cell$ of $\oct$,
	let $r$ be the largest rank of the children of $\cell$ in $\oct$.
	If there is exactly one child of $\cell$ with rank $r$, 
	we set the rank of $\cell$ to $r$. Otherwise, we set the rank
	of $\cell$ to $r+1$.
	
	In the original construction, we insert all
	points in $\cell\setminus \cell'$ to the $I$-projection range tree for $\cell$,
     where
	$\cell'$ is a child of $\cell$ containing the largest number of points of $P$.
	Instead, imagine that we insert all points in $\cell\setminus \cell''$ to the 
	$I$-projection range tree for $\cell''$, where $\cell''$ is a child of $\cell$ with largest rank.
	It is clear that the number of the insertions performed for each internal node
        $\cell$ by this new procedure
        is at least the number of the insertions performed for $\cell$ by the original procedure.
	We give an upper bound on the number of the insertions performed by the new 
	procedure, which proves the lemma.
	
	We claim that each point $p\in P$ is inserted to some 
	$I$-projection range trees for internal nodes at most $O(\log n)$ times during
	the construction.  A cell of $\oct$ has rank at most $\log n$. This
	is because any cell of rank $k$ has at least $2^k$
	descendants. Assume that $p$ is inserted to an
	$I$-projection range tree. Let $\cell_1$ and $\cell_2$ be
	two child nodes (cells) of a cell $\cell$ such that $p\in\cell_1$
	and $p$ is inserted to the $I$-projection range tree for
	$\cell_2$ to form the $I$-projection range tree for their
	parent $\cell$.  There are two cases: the rank of $\cell_1$ is
	smaller than the rank of $\cell_2$, or the rank of $\cell_1$ is
	equal to the rank of $\cell_2$.
	
	In any case, the rank of $\cell$ is larger than the rank of
	$\cell_1$. This means that as you move up a path
	toward the root node of $\oct$, the rank values of the cells
	containing $p$ become larger if $p$ was inserted to
	the $I$-projection range trees of
	the cells. (The rank value remains the same or becomes larger if $p$ was not inserted.) 
	Therefore, the insertion of $p$ occurs at most
	$O(\log n)$ times in total.  Since there are $n$ points to be
	inserted, the total number of insertions is $O(n\log n)$.
\end{proof}

Therefore, we have the following lemma.

\begin{lemma}
	We can construct a data structure of size $O(n\log^{d-1} n)$ in $O(n\log^{d-1} n)$ time
	so that the emptiness of $P\cap Q\cap \cell$ can be checked in
	$O(\log^{d-t-1} n+\log n)$ for any query rectangle $Q$ and any cell $\cell$ of $\oct$ intersecting no $\tl$-face of $Q$ for an integer $t$ with $0< t\leq d$.
\end{lemma}

For a cell $\cell$ containing a corner of $Q$, there is no index $t$ such that
$\cell$ intersects no $\tl$-face of $Q$. 
Thus we simply use the standard
range tree on $P$ and check the emptiness of $P\cap Q\cap \cell$ in $O(\log^{d-1}n)$ time.
Notice that there are $2^d$ cells containing a vertex of $Q$ because the cells are pairwise disjoint.

\subsection{Data Structure for Range-Counting Queries}\label{sec:DS-counting}
The data structure for range-emptiness queries described in Section~\ref{sec:emptiness-DS} can be extended
to a data structure for range-reporting queries. However, it does not seem to work for range-counting queries.
This is because the dynamic range tree with fractional cascading by Mehlhorn and N{\"a}her~\cite{Mehlhorn-1990}
does not seem to support counting queries. Instead, we use a dynamic range tree without
fractional cascading, which increases the query time and update time by a factor
of $\log n$. The other part is the same as the data structure for range-emptiness queries.
Therefore, we have the following lemma.

\begin{lemma}
	We can construct a data structure of size $O(n\log^{d-1} n)$ in $O(n\log^{d-1} n)$ time
	so that the number of points of $P$ contained in $Q\cap \cell$ can be computed in
	$O(\log^{d-t} n+\log n)$ time for any query rectangle $Q$ and any cell $\cell$ of $\oct$ intersecting no $\tl$-face of $Q$ for an integer $t$ with $0< t\leq d$.
\end{lemma}

\section{\texorpdfstring{$k$}{k}-Median Range-Clustering
  Queries} \label{sec:median}
In this section, we present a data
structure and a query algorithm for $k$-median range-clustering
queries.  Given a set $P$ of $n$ points in $d$-dimensional Euclidean space for a
constant $d\geq 2$, our goal is to preprocess $P$ such that $k$-median
range-clustering queries can be answered efficiently. A $k$-median
range-clustering query consists of a $d$-dimensional axis-parallel rectangle 
$Q$, an integer $k$ with $1\leq k\leq n$, and a value $\eps>0$. We
want to find a set $C\in\pset_k$ with
$\median(P_Q,C) \leq (1+\eps) \opt(P_Q)$ efficiently, where
$P_Q=P\cap Q$.  Throughout this section, we use $\Phi$ to denote
$\median$ unless otherwise specified.

Our query algorithm is based on the single-shot algorithm by Har-Peled and
Mazumdar~\cite{Har-Peled-2004}.  
A main difficulty in the implementation for our setting is that they construct a grid
with respect to each point in an approximate center set. Then for each grid cell, they
compute the number of points of $P_Q$ contained in the grid cell.  Thus
to implement their approach in our setting directly, we need to apply
a counting query to each grid cell.  Moreover, we have to avoid
overcounting as a point might be contained in more than one
grid cell of their grid structures.

To do this efficiently without overcounting, we use a \emph{unified grid} based
on the standard quadtree. Although this grid is defined on the standard quadtree, we
use the grid on the compressed quadtree in the implementation.
To make the description easier, we use the standard quadtree
instead of the compressed quadtree in defining of the unified
grid.

\subsection{Coreset Construction from Approximate Centers}
Assume that we are given a
constant-factor approximation $A=\{a_1,\ldots,a_m\}$
  to the $k$-median clustering of $P_Q$, where $m$ is possibly larger than $k$.
In this subsection, we
present a procedure that computes a $(k,\eps)$-coreset of size
$O(|A|\log n/\eps^d)$.  We follow
the approach by Har-Peled and Mazumdar~\cite{Har-Peled-2004}
and implement it in our setting.

\subsubsection{General Strategy}\label{sec:sketch}
We describe our general strategy first, and then show how to
implement this algorithm.  For the definitions of the notations used
in the following, refer to those in Section~\ref{sec:single-shot}
unless they are given. We compute a $2\sqrt{d}$-approximation $R$ to
the maximum of $d(p,A)/ (c_1|P_Q|)$ over all points $p\in P_Q$,
that is, a value $R$ satisfying that the maximum value lies between
$R/(2\sqrt{d})$ and $2\sqrt{d}R$, where
$c_1>1$ is the approximation factor of $A$. Details can be found in Section~\ref{sec:LB}. 

Let $Q_{ij}$ be the cell of the
standard quadtree containing $a_i$ with side length $\tR_{j}$ satisfying
$R2^j\leq \tR_{j} < R2^{j+1}$ for $j=0,\ldots, M=\lceil 2\log (2\sqrt{d}c_1 |P_Q|)\rceil$.
By construction, note that $Q_{ij_1}\subset Q_{ij_2}$ for any two indices $j_1$ and $j_2$
with $j_1 < j_2$.
Note also that for any point $p$ in $P_Q$, we have at least one cell $Q_{ij}$ containing $p$
since there is a value $\tR_j$ at least four times the maximum of $d(p,A)$.

We define the \emph{grid cluster} for $Q_{ij}$ as the union of at most 
$3^d$ grid cells of the standard quadtree with side length $\tR_j$
that share their faces with $Q_{ij}$ including $Q_{ij}$.  
Note that the grid cluster
for $Q_{ij}$ contains all points of $d$-dimensional space that are within
distance from $a_i$ at most $\tR_j$. Also, every point of 
$d$-dimensional space contained in the grid cluster for $Q_{ij}$ has
its distance from $a_i$ at most $2\sqrt{d}\tR_j$.
See Figure~\ref{fig:grid}(b).
Let $V_{i0}$ denote
the grid cluster for $Q_{i0}$ and $V_{ij}$ be the grid cluster for
$Q_{ij}$ excluding the grid cluster for $Q_{i(j-1)}$.  Note that
$V_{ij}$ is the union of at most $3^{d}(2^d-1)$ cells of the standard
quadtree with side length $\tR_j/2$, except for $j=0$.
For $j=0$, the region $V_{i0}$ is the union of at most $3^d$ such cells.

The \emph{first-level grid} for a fixed index $i$ consists of all cells of the standard
quadtree with side length $\tR_j/2$ contained in $V_{ij}$.  For
an illustration, see Figure~\ref{fig:grid}(c).
We partition each cell of
the first-level grid into the cells of the standard quadtree with side
length $\bar{r}_j$ satisfying
$\eps \tR_{j}/(40c_1d) \leq \bar{r}_j \leq 2\eps \tR_{j}/(40c_1d)$.
The \emph{second-level grid} for $i$ consists of all such cells.
Let $\mathcal{V}$ be the set of all
grid cells which contain at least one point of
$P_Q$. Note that the size of $\mathcal{V}$ is $O(|A|\log n/\eps^d)$.
We will see that this set can be obtained in 
$O(|A|\log^{d}n/\eps + |A|\log n/\eps^d)$
time in Section~\ref{sec:Computing the Compressed Cells in the Grid}.

We consider the grid cells $\cell$ of $\mathcal{V}$ one
by one in the increasing order of their side lengths, and
do the followings. Let $P(\cell)$ be the set of points of $P_Q$
that are contained in $\cell$, but are not
contained in any other grid cells we have considered so far.
We compute the number of points of $P(\cell)$, and assign this number
to an arbitrary point of $P(\cell)$ as its weight.
We call this weighted point the \emph{representative} of $\cell$. 
Also, we say that a point of $P(\cell)$ is \emph{charged} to $\cell$.
Notice that every point of $P_Q$ is charged to exactly one cell of $\mathcal{V}$.
We describe the details of this procedure in Section~\ref{sec:Computing the Compressed Cells in the Grid}.
Let $S$ be the set of all such weighted points.

Although the definition of the grid is different from the one by
Har-Peled and Mazumdar~\cite{Har-Peled-2004}, we can still prove that $\mathcal{S}$ 
is a $(k,\eps)$-coreset for $P_Q$ of size
$O(|A|\log n/\eps^d)$ using an argument similar to theirs.

\begin{lemma}
  The set $S$ is a $(k,\eps)$-coreset for $P_Q$ of size $O(|A|\log n/\eps^d)$.	
\end{lemma}
\begin{proof}
  Let $Y$ be an arbitrary set of $k$ points in $d$-dimensional space. 
  For a point $p\in P_Q$, let $\bar{p}$ be the representative
  of the cell which $p$ is charged to.
  Let $\mathcal{E}=|\Phi(P_Q,Y)-\Phi(S,Y)|$. Here, $\Phi(S,Y)$ is the weighted cost function
  between $S$ and $Y$. But we consider $\bar{p}$ as an unweighted point when we
  deal with $d(p,\bar{p})$ and $\dist{\bar{p}}{Y}$.
  By definition, $\mathcal{E} \leq \sum_{p\in P_Q}|\dist{p}{Y}-\dist{\bar{p}}{Y}|$.
  By the triangle inequality, it holds that
  $\dist{p}{Y}\leq \dist{p}{\bar{p}} + \dist{\bar{p}}{Y}$ and $\dist{\bar{p}}{Y}\leq \dist{p}{\bar{p}} + \dist{p}{Y}$ for every point $p$ in $P_Q$, which implies
  $|\dist{p}{Y}-\dist{\bar{p}}{Y}| \leq \dist{p}{\bar{p}}$. 

  Consider a point $p\in P_Q$ such that the cell $\cell$ which $p$ is charged to 
  comes from $V_{i0}$ for some index $i\geq 0$.
  In this case, the side length of $\cell$ is $\bar{r}_0$, which is at most $\frac{2\eps} {40c_1 d} \tR_0 \leq \frac{4\eps }{40c_1 d}R$. Therefore we have
  $\dist{p}{\bar{p}}\leq \frac{4\eps}{40c_1d}R$, and
  the sum of $d(p,\bar{p})$ over all points $p$ in $P_Q$ belonging to this case is at most
  $\frac{4\eps}{40c_1d}R|P_Q|$, which is at most $\frac{4\eps}{40c_1}\Phi(P_Q,A)$ since
  $c_1>1$, $d\geq 2$ and $d(p,A)\leq \Phi(P_Q,A)$ for any $p\in P_Q$.
  
  Now consider a point $p\in P_Q$ such that the cell $\cell$ which $p$ is charged to comes from $V_{ij}$ for any indices $i\geq 0$ and $j >0$.
  Since $j\neq 0$, the distance between $a_i$ and $p$ is at least $\tR_j/4$. 
  The side length of $\cell$ is $\bar{r}_j$, which is at most $\frac{2\eps }{40c_1 d}\tR_j$.
  Therefore, we have 
  $\dist{p}{\bar{p}}\leq \bar{r}_j\leq \frac{8\eps}{40c_1 d}\dist{a_i}{p}$.
  Since we consider the grid cells in $\mathcal{V}$ in the increasing order of their side lengths,
  $p$ is contained in no grid cell of $\mathcal{V}$ of side length at most $\bar{r}_j/2$.
  Therefore, $a_i$ is a constant-factor approximate nearest neighbor of $p$ among the points
  of $A$. More precisely, $\dist{a_i}{p} \leq 2d\cdot \dist{p}{A}$.
  Therefore,
  the sum of $d(p,\bar{p})$ over all points $p$ in $P_Q$ belonging to this case is at most
  $\frac{16\eps}{40c_1} \sum_{p\in P_Q}d(p,A)$, which is $\frac{16\eps}{40c_1}\Phi(P_Q,A)$.
  
  Therefore, we have 
  \[\mathcal{E} \leq \sum_{p\in P_Q} d(p,\bar{p}) \leq 
  \frac{4\eps}{40c_1} \Phi(P_Q,A) + \frac{16\eps}{40c_1}\Phi(P_Q,A) \leq
  \frac{\eps}{c_1}\Phi(P_Q,A)
  \leq \eps \opt(P_Q). \]
  Then, by the definition of $(k,\eps)$-coresets, the lemma holds.
\end{proof}

We implement the algorithm using the compressed quadtree, not the
standard quadtree. We provide an implementation of the algorithm in the following
subsections.

\subsubsection{Computing an Approximation to the Average
  Radius}\label{sec:LB}
The first step is to compute a $2\sqrt{d}$-approximation $R$ to the
maximum of $d(p,A)/(c_1|P_Q|)$ over all points $p\in P_Q$, where
$c_1>1$ is the approximation factor of $A$.  More precisely, we
compute $R$ such that
$R/(2\sqrt{d}) \leq \max_{p\in P_Q} d(p,A)/(c_1|P_Q|) \leq
2\sqrt{d}R$.  We can compute it in $O(|A|\log^{d} n)$ time.

  Let $r^*$ be the maximum of $d(p,A)$ over all points $p\in P_Q$. We
  compute a $2\sqrt{d}$-approximation of $r^*$ and divide it by $c_1|P_Q|$ to compute $R$. Note
  that we can compute $|P_Q|$ in $O(\log^{d-1} n)$ time using the range
  tree constructed on $P$. Imagine that we have a standard grid with side
  length $\alpha>0$ covering $Q$.  Consider the grid cells in this grid
  each of which contains a point of $A$.
  If the union of the grid cluster of all these grid cells contains $P_Q$, it holds that
  $d(p,A)\leq 2\alpha\sqrt{d}$ for any $p\in P_Q$. Otherwise,
  $d(p,A) >\alpha$ for some $p\in P_Q$.  We use this observation to
  check whether $2\alpha\sqrt{d}\geq r^*$ or $\alpha \leq r^*$.
	
  Basically, we apply binary search on the standard lengths.
  However, there are an arbitrarily large number of distinct standard lengths.
  We consider only $O(\log n)$ distinct standard lengths for applying binary search.
  For any value $x$, we use $\sfloor{x}$ and $\sceil{x}$ to denote
  the largest standard length which is smaller than or equal to $x$, and the smallest standard
  length which is larger than or equal to $x$, respectively.
  The following lemma is used for a subprocedure in the binary search.
  
\begin{lemma}\label{lem:subproc}
  Given a standard length $\alpha$, we can check whether $\alpha$
  is at most $r^*$ or at least
  $r^*/(\alpha \sqrt{d})$ in $O(|A|\log^{d-1} n)$ time. 
\end{lemma}  
\begin{proof}
  We find the cells of the standard quadtree with side length $\alpha$ that
  contain $a$ in their grid clusters for each point $a$ of $A$.
  The union $U$ of all these grid clusters consists of $3^d|A|$ cells of $\octs$
  with side length $\alpha$.
  We want to check whether every point of $P_Q$ is contained in $U$.
  If so, $r^*$ is at most $2\alpha\sqrt{d}$. Otherwise, $r^*$ is at least $\alpha$.
  To do this, for each cell $\cell$ with side length $\alpha$ contained in $U$, we compute
  the number $N(\cell)$ of points of $P\cap Q$ that are contained in $\cell$ in $O(\log^{d-1}n)$ time
  using the range tree on $P$.
  Since the cells are pairwise interior disjoint, the sum of $N(\cell)$ of all cells $\cell$
  is $|P_Q|$ if and only if all points of $P_Q$ are in the union of all such cells.
  Therefore, we can check whether all points of $P_Q$ are in $U$ in $O(|A|\log^{d-1} n)$ time.
\end{proof}

We apply binary search on a set $\mathcal{L}$ of standard lengths
defined as follows.  For every pair $(p,a)$ with $p\in P$ and
$a\in A$, consider the difference $\ell$ between the $i$th
  coordinates of $p$ and $a$ for every $1\leq i\leq d$.  Let
$\mathcal{L}$ be the sorted list of $\sfloor{\ell}$ for every difference $\ell$. The size
of $\mathcal{L}$ is $d |A|n$. Imagine that we have the sorted lists
of $\mathcal{L}$.  For every iteration, we choose the median $\alpha$
of the search space of $\mathcal{L}$ and check if
$\alpha\geq r^*/2\sqrt{d}$ or $\alpha \leq r^*$. If
  $\alpha\geq r^*/(2\sqrt{d})$, we consider the lengths smaller than
  $\alpha$ in the current search space for the next iteration.
Otherwise, we consider the lengths larger than $\alpha$ for the next
iteration.  In this way, we obtain an interval $[\alpha_L,\alpha_U]$
satisfying that either $\alpha_L \leq r^*$ and
$a_U\geq r^*/(2\sqrt{d})$ or $r^*/(2\sqrt{d})\leq \alpha_U\leq r^*$ in
$O(|A|\log^d n)$ time in total. We return $\alpha_U$ as an output.
The following lemma shows that this binary search can be done in the
same time without computing $\mathcal{L}$ explicitly.

\begin{lemma}
  We can compute $\alpha_U$ in $O(|A|\log^d n)$ time after
  an $O(n\log n)$-time preprocessing on $P$.
\end{lemma}  
\begin{proof}
  We apply binary search on $\mathcal{L}$ without computing it explicitly.
  As a preprocessing, we compute a balanced binary search tree on the projection of $P$ onto each axis.
  We have $d$ binary search trees, and we can compute them in $O(n\log n)$ time. This time is subsumed
  by the total construction time. 
  
  For the binary search, we locate every point of $A$ in the balanced
  binary search trees in $O(|A|\log n)$ time in total. Then we have two search spaces for each pair $(a,i)$ with $a\in A$ and $1\leq i\leq d$: the differences of the $i$th coordinates of $a$ and
  the points of $P$ lying on the $i$th axis in one direction from $a$, and the difference
  of the $i$th coordinates of $a$ and the points of $P$ lying on the $i$th axis
  in the other direction from $a$.
  (Precisely, we apply apply $\sfloor{\cdot}$ operation to each element.)
  
  We can apply binary search on each search space using the balanced binary search trees. Note that we have $O(|A|)$ search spaces. To accomplish the task more efficiently,
  we apply binary search on all search spaces together as follows.
  We choose the median for each search space, and assign the size of the search space to the median
  as its weight in $O(|A|\log n)$ time in total.
  Then we choose the weighted median $\alpha$ of the weighted medians in $O(|A|)$ time.
  Then we test whether $\alpha\geq r^*/2\sqrt{d}$ or $\alpha \leq r^*$ in $O(|A|\log^{d-1}n)$ time
  by Lemma~\ref{lem:subproc}. 
  Regardless of the result,
  the size of the total search space decreases by a constant factor.
  Therefore, in $O(\log n)$ iterations, we can obtain a desired interval in $O(|A|\log^d n)$ time in total.
\end{proof}

\begin{lemma}
  	The standard length $\alpha_U$ is a $2\sqrt{d}$-approximation to $r^*$.
\end{lemma}
\begin{proof}
  We already showed that the interval $[\alpha_L, \alpha_U]$ satisfies
  one of the following conditions: either $\alpha_L \leq r^*$ and
  $a_U\geq r^*/(2\sqrt{d})$ or $r^*/(2\sqrt{d})\leq \alpha_U\leq r^*$.
  For the latter case, the lemma holds immediately. If there is at
  least one standard length in $\mathcal{L}$ lies between
  $r^*/(2\sqrt{d})$ and $r^*$, the output interval belongs to the
  latter case by construction.  Thus assume there is no such standard
  length in $\mathcal{L}$, and $[\alpha_L,\alpha_U]$ belongs to the
  former case.
  
  Let $(p,a)$ be a pair with $p\in P_Q$ and $a\in A$ such that
  $d(p,a)$ is the maximum $r^*$ of $d(p,A)$ for all points of $p$. Let $i$ be an integer
  with $1\leq i\leq d$ that maximizes the length $\ell$
  of the projection of the segment $\overline{pa}$ onto the $i$th axis.
  We have $r^*/\sqrt{d}\leq \ell\leq r^*$.
  By the construction, $\alpha=\sfloor{\ell}$ is in $\mathcal{L}$.
  We have
  $r^*/(2\sqrt{d}) \leq \alpha \leq r^*$. This contradicts the assumption
  that no standard length of $\mathcal{L}$ lying between $r^*/(2\sqrt{d})$ and $r^*$.
  Therefore, $\alpha_U$ is a $2\sqrt{d}$-approximation to
  $r^*$, and the lemma holds.
\end{proof}

\begin{lemma}
  We can compute a $2\sqrt{d}$-approximation to the maximum of
  $d(p,A)/(c_1|P_Q|)$ for all points $p$ in $P_Q$ in $O(|A|\log^d n)$
  time.
\end{lemma}
\subsubsection{Computing the Compressed Cells in the
  Grid}\label{sec:Computing the Compressed Cells in the Grid}
As described in Section~\ref{sec:sketch}, we construct the
second-level grid for each index $i$ for $i=1,\ldots,m$, and check
whether each grid cell contains a point of $P_Q$. The set of the grid
cells in the second-level grids containing a point of $P_Q$ is denoted
by $\mathcal{V}$.  Then we consider the grid cells $\cell$ of
$\mathcal{V}$ one by one in the increasing order of their side
lengths, and compute the number of points of $P_Q$ contained in
$\cell$, but not contained in any other grid cells we have considered
so far. Computing this number is quite tricky.

To handle this problem, we observe that for any two cells in
$\mathcal{V}$, either they are disjoint or one is contained in the
other.  This is because they are cells of the standard quadtree.
For two cells $\cell_1$ and $\cell_2$ with $\cell_1\subseteq \cell_2$, let
$i_1$ and $i_2$ be the indices such that $\cell_1$ and $\cell_2$ are
grid cells of the second-level grids for $i_1$ and $i_2$,
respectively.  Since the grid cells in the same second-level grid
are pairwise interior disjoint, we have $i_1\neq i_2$. In this
case, for any point $p\in \cell_2$, 
there is another grid cell
$\cell_1'$ containing $p$ in the second-level grid for $i_1$ with side length smaller
than the side length of $\cell_2$.
Therefore, we do not consider any
cell of $\mathcal{V}$ containing another cell of $\mathcal{V}$.
Imagine that we remove all such cells from $\mathcal{V}$. Then the
cells of $\mathcal{V}$ are pairwise interior disjoint. Therefore, if
suffices to compute the number of points of $P_Q$ contained in each
cell of $\mathcal{V}$, which can be done efficiently using the data
structure described in Section~\ref{sec:DS}.

In the following, we show how to compute the set $\mathcal{V}$ after
removing all cells containing another cell efficiently. To do this, we
first compute the cells in the first-level grids, and discard some of
them. Then we subdivide the remaining cells into cells in the
second-level grids. More specifically, let $\mathcal{V}_1$ be the set
of the cells of the first-level grids. We
first compute the cells in $\mathcal{V}_1$, and then remove all cells
in $\mathcal{V}_1$ containing another cell in
$\mathcal{V}_1$.  Then the cells in $\mathcal{V}_1$ are pairwise
interior disjoint.  And then we compute the second-level grid cells in
each cell of $\mathcal{V}_1$.  The second-level grid cells we
obtain are the cells of $\mathcal{V}$ containing no other cell in
$\mathcal{V}$.  Also, in the following, to apply
Lemma~\ref{lem:counting-query}, we consider the compressed cells
instead of the cells in the standard quadtree.

\subparagraph{First-Level Grid.}
We compute the cells of the first-level grid for every index $i$.
There are $O(|A|\log n)$ cells of the first-level grids in total.
We compute them in $O(|A|\log n)$
time and compute the compressed cell for each cell in
$O(|A|\log^2 n)$ time in total by Lemma~\ref{lem:cell-query}.  We
remove all compressed cells containing another compressed
cells in $O(|A|\log^2 n)$ time using the following lemma.

\begin{lemma}
  We can find all compressed cells of the cells of the first-level
  grids containing another compressed cells in $O(|A|\log^2 n)$ time
  in total.
\end{lemma}
\begin{proof}
  Recall that a cell of the compressed quadtree can be represented as
  an interval using the $\mathcal{Z}$-order. The description of this order
  is given in Section~\ref{sec:quad} of Appendix.  Let
  $\langle\ccell_1,\ldots,\ccell_{k'}\rangle$ be the sequence of the
  compressed cells of the cells of the first-level grids in the
  increasing order of their side lengths.  For each index $t$ with
  $1\leq t\leq k'$, we check whether there is an index $t'<t$ with
  $\ccell_{t'}\subseteq \ccell_t$. To do this, we consider the cells
  from $\ccell_1$ to $\ccell_{k'}$ and maintain an interval tree
  $\mathcal{I}$. The interval tree contains all intervals corresponding to the cells
  we have considered so far. Since the sequence of the insertions to the interval
  tree is known in advance, each insertion can be done in $O(\log n)$ time.  
  
  To check whether there is an index $t'$ with
  $\ccell_{t'}\subseteq \ccell_t$ for some index $t$, we check whether
  the interval corresponding to $\ccell_t$ contains another interval
  in $\mathcal{I}$. This can be done in $O(\log n)$ time. Since there are
  $O(|A|\log n)$ cells of the first-level grids, we can find all
  compressed cells of the cells of the first-level grids contained in
  another compressed cells in $O(|A|\log^2 n)$ time in total.
\end{proof}

The resulting grid cells are pairwise disjoint and contain $P_Q$ in their union.
But it is possible that a grid cell does not contain a point of $P_Q$.

\subparagraph{Second-level Grids.}  For each compressed cell $\ccell$
of the cells of the first-level grids, we compute the second-level
grids constructed from it. To do this, we traverse the subtree of
$\ccell$ of $\oct$ towards its leaf nodes. More precisely, let
$\mathcal{V}$ be the singleton set containing $\ccell$.  We pick the
largest cell of $\mathcal{V}$, and insert its children to
$\mathcal{V}$.  We do this until the largest cell of $\mathcal{V}$ has
side length at most $\bar{r}_j$ assuming that $\ccell$ comes from a
grid cluster $V_{ij}$.  Notice that some of them may not
  intersect $Q$. This takes time linear in the number of the grid cells
  in the second-level grids.

\subparagraph{Range-Counting for Each Compressed Cell.}  The next step
is to compute the number of points of $P_Q$ contained in each
cell $\ccell$ in the second grids.  If $\ccell$ is contained in $Q$, we already
have the number of points of $P_Q$ contained in $\ccell$, which is
computed in the preprocessing phase.  If $\ccell$ contains a corner of
$Q$, we use the range tree constructed on $P$.  Since there are $O(1)$
such cells, we can handle them in $O(\log^{d-1}n)$ time in total.  For
the remaining cells, we use the data structure in
Section~\ref{sec:DS-counting}.  Then we can handle them in
$O(\sum_{t=1}^{d-1} m_t \log^{d-t} n)$ time, where $m_t$ is the number
of the cells of $\mathcal{V}$ intersecting no $\tl$-face of $Q$ but
intersecting a $t$-dimensional face of $Q$ for an integer with
$0<t< d$.  We have $m_t=O(|A|\log n/\eps^{t})$.  Therefore, the total
running time for the range-counting queries is
$O(|A|\log^2 n/\eps^{d-1} + |A|\log^d n/\eps + \log^{d-1} n + |A|\log
n/\eps^d )$ in total, which is
$O(|A|\log^d n/\eps + |A|\log n/\eps^d )$.

Therefore, we have the following lemma.
\begin{lemma}\label{lem:approx-center-to-coreset}
  Given a constant-factor approximation $A$ to the $k$-median clustering
    of a set $P$ of
  $n$ points in $d$-dimensional space such that $|A|$ is possibly larger than $k$,
  we can compute a $(k,\eps)$-coreset of $P_Q$ of size $O(|A|\log n/\eps^d)$
  in $O(|A|\log^d n/\eps + |A|\log n/\eps^d )$ time
  for any rectangle $Q$, any integer $k$ with $1\leq k\leq n$ and any value $\eps>0$.
\end{lemma}

\subsection{Smaller Coreset}
Due to Lemma~\ref{lem:constant-coreset}, we can obtain a
$(k,2)$-coreset $S$ of $P_Q$ of size $O(k\log^{d} n)$ in
$O(k\log^{d} n)$ time for any query rectangle $Q$ using a data structure
of size $O(n\log^d n)$. 
A $(k,c)$-coreset of $S$ is also a $(k,2c)$-coreset of $P_Q$ for any constant $c>1$ 
by the definition of the coreset. We compute a $(k,2)$-corset $S'$ of $S$,
which is a $(k,4)$-coreset of $P_Q$,   
of size $O(k\log n)$ in $O(k\log^{d} n+k^5 \log^9 n)$ time
by~\cite[Lemma 5.1]{Har-Peled-2004} by setting
$\eps=2$.

Using this $(k,4)$-coreset of size $O(k\log n)$ of $P_Q$, we can obtain 
constant-factor approximate centers of size $k$ as Har-Peled and
Mazumdar~\cite{Har-Peled-2004} do. 
We compute a constant-factor $k$-center clustering $\pset_0$ of the coreset using~\cite{Gon1985}.
Then we apply the local search
algorithm due to Arya et al.~\cite{Arya-2004} to $\pset_0$ and $S'$ to
obtain a constant-factor approximation to $\opt(S)$. This
takes $O(|S'|^2 k^3 \log n)=O(k^5\log^3 n)$ time, and finally $\mathcal{C}_0$
becomes a constant-factor approximation to $\opt(S)$ of size $k$~\cite{Har-Peled-2004}.
Therefore, we can compute a $(k,\eps)$-coreset of size
$O(k\log n/\eps^d)$ using Lemma~\ref{lem:approx-center-to-coreset}
using the constant-factor approximation $\mathcal{C}_0$ to $\opt(S)$ of size $k$.

\begin{lemma}
  Given a query range $Q\subseteq\mathbb{R}^d$, an integer $k$ with $1\leq k\leq n$, 
  and a value $\eps>0$ as a query, we can compute a $(k,\eps)$-coreset
  of $P_Q$ for the $k$-median range-clustering of size $O(k\log n/\eps^d)$
  in $O(k^5\log^9 n+ k\log^d n/\eps + k\log n/\eps^d )$
  time.
\end{lemma}

\begin{theorem}
	Let $P$ be a set of $n$ points in $d$-dimensional space. There is a
	data structure of size $O(n\log^d n)$ such that given a query
	range $Q\subseteq\mathbb{R}^d$, an integer $k$ with $1\leq k\leq n$,
	and a value $\eps>0$ as a query, an
	$(1+\eps)$-approximation to the $k$-median range-clustering of
	$P\cap Q$ can be computed in $O(k^5\log^9 n+k\log^d /\eps +T_\textnormal{ss}(k\log n/\eps^d))$ time,
	where $T_{\textnormal{ss}}(N)$ denotes the running time of an $(1+\eps)$-approximation
	single-shot algorithm for the $k$-median clustering of $N$ weighted input points.
\end{theorem}

If we use the algorithm in~\cite{Har-Peled-2004} for computing an
$(1+\eps)$-approximation to the $k$-median clustering, 
$T_\textnormal{ss}(N)=O(N\log^2 W+k^5\log^9 W+ \varrho k^2\log^5 W)$,
where 
$W$ is the total weight of the input points and 
$\varrho =\exp[O((1+\log (1/\eps))/\eps)^{d-1}]$.
Therefore, we have the following corollary.
In the running time of the corollary, the term of $k\log n/\eps^d$
is subsumed by the term of $\varrho k^2\log^5 n$.

\begin{corollary}
  Let $P$ be a set of $n$ points in $d$-dimensional space. There is a
    data structure of size $O(n\log^d n)$ such that given a query
  range $Q\subseteq\mathbb{R}^d$, an integer $k$ with $1\leq k\leq n$,
  and a value $\eps>0$ as a query, an
  $(1+\eps)$-approximation to the $k$-median range-clustering of
  $P\cap Q$ can be computed in
  $O(\varrho k^2\log^5 n+ k^5\log^9 n+ k\log^d n/\eps)$ time, where
  $\varrho =\exp[O((1+\log (1/\eps))/\eps)^{d-1}]$.
\end{corollary}

\subparagraph{Remark.}  The construction of the coreset for the
$k$-means clustering is similar to the construction of the coreset for
the $k$-median clustering in~\cite{Har-Peled-2004}.  The only
difference is that for the $k$-means clustering
$\means$ is used instead of $\median$ and
$R=\sqrt{\means (P,A)/(c_1n)}$ is used instead of $R=\median(P,A)/(c_1n)$.
Therefore, we can compute a
$(k,\eps)$-coreset for the $k$-means clustering of size
$O(k\log n/\eps^d)$ in
$O(k^5\log^9 n+ k\log^{d} n/\eps + k\log n/\eps^d)$ time. 

\begin{theorem}
	Let $P$ be a set of $n$ points in $d$-dimensional space. There is a
	data structure of size $O(n\log^d n)$ such that given a query
	range $Q\subseteq\mathbb{R}^d$, an integer $k$ with $1\leq k\leq n$,
	and a value $\eps>0$ as a query, an
	$(1+\eps)$-approximation to the $k$-means range-clustering of
	$P\cap Q$ can be computed in $O(k^5\log^9 n+k\log^d /\eps +T_\textnormal{ss}(k\log n/\eps^d))$ time,
	where $T_{\textnormal{ss}}(N)$ denotes the running time of an $(1+\eps)$-approximation
	single-shot algorithm for the $k$-means clustering of $N$ weighted input points.
\end{theorem}

 Since an
$(1+\eps)$-approximate $k$-means clustering of $N$ weighted points
of total weight $W$ can be computed in
$O(N\log^2 W+k^5 n\log^5 W + k^{k+2} \eps^{(-2d+1)k} \log^{k+1}
W\log^k (1/\eps))$ time~\cite{Har-Peled-2004}, we have the following corollary.

\begin{corollary}
	Let $P$ be a set of $n$ points in $d$-dimensional space. There is a
	data structure of size $O(n\log^d n)$ such that given a query
	range $Q\subseteq\mathbb{R}^d$, an integer $k$ with $1\leq k\leq n$,
	and a value $\eps>0$ as a query, an
	$(1+\eps)$-approximation to the $k$-median range-clustering of
	$P\cap Q$ can be computed in
	$O(k^6\log^6 n/\eps^d + k^{k+2} \eps^{(-2d+1)k} \log^{k+1} n\log^k (1/\eps) +
          k^5\log^9 n+ k\log^{d} n/\eps)$ time.
\end{corollary}

\section{\texorpdfstring{$k$}{k}-Center Range-Clustering Queries}
In this section, we are given a set $P$ of $n$ points in
$d$-dimensional Euclidean space for a constant $d\geq 2$.  Our goal is to process $P$
so that $k$-center range-clustering queries can be computed efficiently. A
range-clustering query consists of a rectangle
$Q\subseteq\mathbb{R}^d$, an integer $k$ with $1\leq k\leq n$, and a value
$\eps>0$. We
want to find a set $C\in\pset_k$ with
$\ccenter(P_Q,C) \leq (1+\eps) \opt(P_Q)$ efficiently, where $P_Q=P\cap Q$.  In this
section, we use $\Phi$ to denote $\ccenter$.

\subparagraph{Sketch of the Algorithm by Abrahamsen et al.}
Abrahamsen et al.~\cite{Abrahamsen-2017} present a data structure and
its query algorithm for this problem. They 
construct a compressed quadtree on $P$ as a data structure.
Their query algorithm consists of two phases. In the first phase, they
compute a lower bound $\lb$ of $\opt(P_Q)$,
and then obtain a set of $O(k)$ pairwise interior disjoint cells of $\oct$ with side length
at most $\lb$ whose union contains all points of $P_Q$.  In the second
phase, they subdivide the cells they obtained 
so that the side length of a cell becomes at most $\eps\lb$. Then for
each cell that contains a point of $P\cap Q$, they choose an arbitrary
point in the cell. They use the set of all chosen points as a
$(k,\eps)$-coreset. By applying a single-shot algorithm for
the $k$-center clustering to the coreset,
they can obtain an $(1+\eps)$-approximate $k$-center range-clustering. 
The first phase takes $O(k\log^{d-1}n )$ time, and
the second phase takes $O(k(\log n/\eps)^{d-1})$ time.
In this section, we show that the second phase can be done in $O(k\log^{d-1} n + k/\eps^d)$ time using
the data structure described in Section~\ref{sec:DS}.

\subparagraph{Data Structure.}
We construct a compressed quadtree $\oct$ on $P$.  For each cell $\cell$ of $\oct$,
we store the point of $P\cap \cell$ closest to each facet of $\cell$.
Also, we mark whether or not $\cell$ contains a point of $P$.  
Due to this information, given the node of $\oct$ corresponding to a cell $\cell$,
we can check whether $P_Q\cap\cell$ is empty or not 
in constant time if $\cell$ crosses only one facet of $Q$ or it is contained in $Q$.
Also, we construct $I$-projection range trees on $P$
described in Section~\ref{sec:DS}.
The total space complexity is $O(n\log^{d-1} n)$.

\subparagraph{Query Algorithm.}
We are given a query rectangle $Q$, an integer $k$ and a value $\eps$. Also, assume that
we have the cells obtained from the first phase of the algorithm by Abrahamsen et al.~\cite{Abrahamsen-2017}.
For each cell $\cell$ obtained
from the first phase, we traverse the subtree of $\cell$ of $\oct$ 
towards its leaf nodes until we reach the cells with side length at most $\eps\lb$. 
More precisely, let $\mathcal{G}(\cell)$ be a set of descendants of $\cell$ in $\oct$,
which is initially set to the singleton set containing $\cell$.
We remove the largest cell of $\mathcal{G}(\cell)$ from $\mathcal{G}(\cell)$, and
insert its children to $\mathcal{G}(\cell)$. 
We do this until the largest cell of $\mathcal{G}(\cell)$ has side length at most $\eps\lb$.
Then we remove the cells of $\mathcal{G}(\cell)$ not intersecting $Q$ from $\mathcal{G}(\cell)$. 
The union $\mathcal{G}$ of 
all $\mathcal{G}(\cell)$'s is the set of all cells they obtain in the second phase. 
This takes time linear to the number of cells in $\mathcal{G}$, which is $O(k/\eps^d)$.

Each cell $\cell$ of $\mathcal{G}$ belongs to one of the three types:
$\cell$ is contained in $Q$, $\cell$ contains a corner of $Q$, 
and otherwise.  We want to check whether or not each cell $\cell$ of $\mathcal{G}$ 
contains a point of $P\cap Q$.  
For a cell of the first type, we can check this in $O(1)$ time using the information
stored in $\cell$. There are
$O(k/\eps^d)$ cells of the first type.
For a cell $\cell$ of the
second type, we use the range tree on $P$ and check the emptiness in $O(\log^{d-1} n)$ time.
There are $O(1)$ cells of the third type.
For a cell of the there type, there is an integer $t$ with $0<t<d$ such that $\cell$
intersects no $\tl$-face of $Q$ but intersects a $t$-dimensional face of $Q$.
There are $O(k/\eps^t)$ cells of $\mathcal{G}$ intersecting no $\tl$-face of $Q$ but
intersecting a $t$-dimensional face of $Q$. 
Therefore, the cells of the fourth type can be handled in
$O(k\sum_{t=1}^{d-1} (\log^{d-t-1}n+\log n)/\eps^{t})=O(k\log^{d-2} n/ \eps +
k\log n/\eps^{d-1})$ time in total.

The overall running time is 
$O(k\log^{d-1} n+k/\eps^d + k\log^{d-2} n/ \eps + k\log n/\eps^{d-1})$, which is
$O(k\log^{d-1} n+k/\eps^d+k\log n/\eps^{d-1})$.
Therefore, we have the following lemma.
The paper~\cite{Abrahamsen-2017} deals with a more general cost function which they call a \emph{$(c,f(k))$-regular function}. For definition, see Definition~1 of~\cite{Abrahamsen-2017}.
The method in this section can be directly applied for the $(c,f(k))$-regular function.

\begin{lemma}
  Given any query range $Q$, an integer $k$ with $1\leq k\leq n$,
  and a value $\eps>0$ as a query, we can
  compute a $(k,\eps)$-coreset of $P_Q$ for the $k$-center range-clustering
  of size $O(k/\eps^d)$ in
  $O(k\log^{d-1} n+ k/\eps^d+k\log n/\eps^{d-1})$ time 
  using a data
  structure of size $O(n\log^{d-1} n)$.
\end{lemma}

\begin{theorem}
	Let $P$ be a set of $n$ points in $d$-dimensional Euclidean space. There is a
	data structure of size $O(n\log^{d-1} n)$ such that given a query
	range $Q\subseteq\mathbb{R}^d$, an integer $k$ with $1\leq k\leq n$,
	and a value $\eps>0$ as a query, an
	$(1+\eps)$-approximation to the $k$-center range-clustering of
	$P\cap Q$ can be computed in $O(k\log^{d-1}n +k\log n/\eps^{d-1}+T_\textnormal{ss} (k/\eps^d))$ time, where $T_{\textnormal{ss}}(N)$ denotes the running time of an $(1+\eps)$-approximation
	single-shot algorithm for the $k$-center clustering of $N$ input points.
\end{theorem}

The algorithm by Agarwal and Cecillia computes the exact $k$-center clustering of $N$ points
in $d$-dimensional space under any $L_p$-metric in $N^{O(k^{1-1/d})}$ time~\cite{agarwal2002}.

\begin{corollary}
	Let $P$ be a set of $n$ points in $d$-dimensional Euclidean space. There is a
	data structure of size $O(n\log^{d-1} n)$ such that given a query
	range $Q\subseteq\mathbb{R}^d$, an integer $k$ with $1\leq k\leq n$,
	and a value $\eps>0$ as a query, an
	$(1+\eps)$-approximation to the $k$-center range-clustering of
	$P\cap Q$ can be computed in $O(k\log^{d-1}n+k\log n/\eps^{d-1} + (k/\eps^d)^{O(k^{1-1/d})})$ time. 
\end{corollary}

\section{Approximate Diameter and Radius of a Point Set}
In this section, we are given a set $P$ of $n$ points in $d$-dimensional Euclidean space.
Our goal in this section is to preprocess $P$ so that given any orthogonal range $Q$
and a value $\eps>0$, an approximate diameter (or radius) of $P \cap Q$
can be computed efficiently.
This problem can be considered as a special case of the clustering problem where the
number of clusters is only one.

This problem was studied by Gupta et al.~\cite{Gupta} and Nekrich and
Smid~\cite{Nekrich-2010}.  Gupta et al.~\cite{Gupta} considered this
problem in the plane and presented two data structures.
One requires $O(n\log^2 n)$ size that supports queries
with arbitrary approximation factors $1+\eps$ in $O(\log n/\sqrt{\eps}+\log^3 n)$
query time and the other requires a smaller size $O(n\log n/\sqrt{\delta})$ that
supports only queries with the \textit{fixed} approximation factor $1+\delta$
with $0<\delta<1$ that is used for constructing the data structure.
Later, Nekrich and Smid presented a data structure for this problem in a
higher dimensional space that has size $O(n \log^d n)$ and supports
diameter (or radius) queries with the fixed approximation factor  
$1+\delta$ in $O(\log^{d-1} n/\delta^{d-1})$ query time.
Here, $\delta$  is the approximation
factor given for the construction of their data structure, and
therefore it is fixed for queries to the data structure. That is,
the data structure does not support any queries with approximation
factors other than $(1+\delta)$.

We present data structures and a query algorithm
  for this problem. In the plane, our data structure requires $O(n\log n)$ size and
  supports diameter (or radius) queries with arbitrary approximation
  factors $1+\eps$ in $O(\log n/\eps)$ query time.
  In higher dimension $d$, our data structures not only allow
  queries to have arbitrary approximation factor values $1+\eps$, but
  also improve the size and the query time of the data structure.
  The size is improved by a factor of $\log n$.
  Even when $\eps$ is fixed to $\delta$, the query time is improved
  by a factor of $\min\{1/\delta^{d-1}, \log^{d-2} n\}$.


\subparagraph{$\eps$-Coresets.} Our query algorithm starts by sampling
a set $S$ of points from $P\cap Q$, which we call an $\eps$-coreset of
$P\cap Q$, such that the diameter of $S$ is an
$(1+\eps)$-approximation of the diameter of $P\cap Q$. Let $\apx$ be a
value such that $D \leq \apx \leq c\cdot D$ for a constant $c>1$, where
$D$ is the diameter of $P\cap Q$.
Consider a standard grid of side
length $\eps\apx$ covering $Q$. Assume that we pick an arbitrary point
in each grid cell containing a point of $P\cap Q$. Then the set of all picked points
is an $\eps$-coreset of $P\cap Q$ of
size $O(1/\eps^d)$.  Let $\mathcal{D}$ be the set of all grid cells containing
a point of $P\cap Q$.

We can obtain a smaller $\eps$-coreset as follows.
We first obtain a subset $\mathcal{D}'\subseteq \mathcal{D}$
and choose an arbitrary point
in each grid cell of $\mathcal{D}'$ for a $\eps$-coreset. If a grid cell of $\mathcal{D}$
intersects the boundary of $Q$, we move it from $\mathcal{D}$ to $\mathcal{D}'$.
The remaining cells are contained in $Q$. 
For the remaining cells of $\mathcal{D}$, 
consider the grid cells of $\mathcal{D}$  whose centers have the same
coordinates, except for only one coordinate, say the $i$th coordinate.  We add the grid
cells with the largest $i$th coordinate and smallest $i$th coordinate to
$\mathcal{D}'$. See Figure~\ref{fig:coreset-diam}. 
Then $\mathcal{D}'$ consists of $O(1/\eps^{d-1})$ grid cells.
We choose an arbitrary point of $P$ contained in each grid cell of $\mathcal{D}'$.
The set $S$ of all chosen points is an $\eps$-coreset of $P\cap Q$ of size
$O(1/\eps^{d-1})$. 

\begin{figure}
	\begin{center}
		\includegraphics[width=0.3\textwidth]{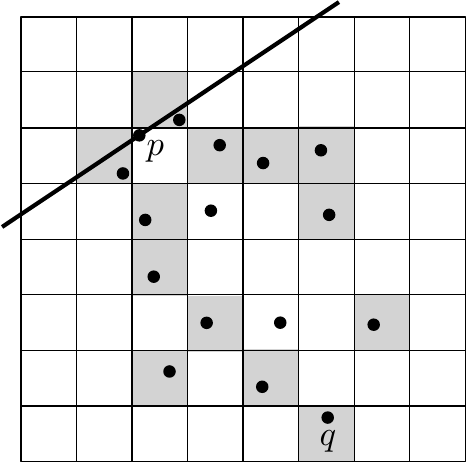}
		\caption{\small The gray cells are chosen as a $\eps$-corset of size $O(1/\eps^{d-1})$. \label{fig:coreset-diam}}
	\end{center}
\end{figure}

\begin{lemma}
	The set $S$ is an $\eps$-coreset of $P\cap Q$.
\end{lemma}
\begin{proof}
	Let $p$ and $q$ be two points of $P\cap Q$ such that $d(p,q)$ is the diameter of $P\cap Q$.
	Consider the hyperplane $h$ orthogonal to $\overline{pq}$ and passing through $p$.
	A (closed) half-space bounded by $h$ contains all points of $P\cap Q$, and the other (open)
	half-space $\mathcal{H}$ bounded by $h$ contains no point of $P\cap Q$. See Figure~\ref{fig:coreset-diam}.
	There is a ray starting from $p$ and parallel to an axis that does not intersect $\mathcal{H}$.
	This means that a grid cell in the grid cluster
	of the grid cell containing $p$ is chosen by the construction,
	and the same holds for $q$. Let $p'$ and $q'$ be the chosen points from the cells containing 
	$p$ and $q$, respectively. 
	The diameter of $S$ is at most $d(p',q')$, and we have
	$d(p',q') \leq d(p,q)+2\eps \apx \leq D+\eps c\cdot D \leq (1+c') D$ for a constant $c'$.
	Therefore, $S$ is an $\eps$-coreset of $P\cap Q$.
\end{proof}

\subparagraph{Data Structure.}  We construct the standard range tree
on $P$, the compressed quadtree on $P$, and the data structure on $P$
described in Section~\ref{sec:DS}.
We also maintain another data structure similar to the one described
in Section~\ref{sec:DS}. For each index $i$, we project all points of
$P$ onto a hyperplane orthogonal to the $i$th axis. Let $P_i$ be the
set of such projected points.  If there is more than one point of $P$
which is projected to the same point, we consider them as distinct points
in $P_i$.  We construct the compressed quadtree $\oct(P_i)$ on $P_i$
that is aligned to $\oct$. Given a cell $\cell$ of $\oct(P_i)$ and a
value $x$, we want to find the point with the largest $i$th coordinate
smaller than $x$ among the points of $P_i\cap \cell$.  To do this, we
do the followings.  For each point $p_i$ in a cell $\cell_i$ of
$\oct(P_i)$, there is a unique point $p$ in $P$ whose projection
onto the hyperplane orthogonal to the $i$th axis is $p_i$.
We assign the $i$th coordinate of $p$ to $p_i$ as its
weight.  Then we compute the 1-dimensional range tree (balanced binary
search tree) on $P_i\cap \cell_i$ with respect to their weights for
each cell $\cell_i$ of $\oct(P_i)$.  As we did in
Section~\ref{sec:DS}, we use a persistent data structure instead of
computing the balanced binary search tree explicitly.  Since we
consider the balanced binary search tree here, we do not need to use
fractional cascading.  Each insertion takes $O(\log n)$ time.
Therefore, this data structure has size $O(n\log^{d-1} n)$ and can
be constructed in $O(n\log^{d-1} n)$ time in total.  For every cell
$\cell$ of $\oct$ and an index $i$, there is a cell $\cell_i$ of
$\oct(P_i)$ such that the projection of $\cell$ onto the hyperplane
orthogonal to the $i$th axis is $\cell_i$. We
make $\cell$ to point to $\cell_i$.

\subparagraph{Query Algorithm.}  We are given an orthogonal range $Q$
and a value $\eps>0$ as a query. We first compute a constant-factor
approximation $\apx$ to the diameter of $P\cap Q$ in $O(\log^{d-1}n)$
time using the standard range tree. To do this, for each facet of $Q$,
we find the points of $P\cap Q$ closest to the facet in
$O(\log^{d-1} n)$ time. That is, we compute the smallest enclosing box $\meb$ of 
$P\cap Q$. The diameter $\apx$ of $\meb$ is a
constant-factor approximation to the diameter of $P\cap Q$.
Assume that $\eps\apx$ is a standard length. Otherwise, we consider
the largest standard length smaller than $\eps\apx$ instead of $\eps\apx$.

Then we compute an $\eps$-coreset of $P\cap Q$ of size $O(1/\eps^{d-1})$ as follows.
Consider the standard grid with side length $\eps\apx$ covering $\meb$.
Here, we do not compute this grid
explicitly because there are $O(1/\eps^d)$ cells in this grid. Instead, we
compute the grid cells intersecting the boundary of $\meb$.  There are
$O(1/\eps^{d-1})$ such cells.  For each such cell $\cell$, we check
whether or not $\cell$ contains a point of $P\cap Q$ using the data
structure in Section~\ref{sec:DS}.  There are $O(1/\eps^t)$
cells intersecting no $\tl$-face of $\meb$ but intersecting
a $t$-dimensional face of $\meb$ for an integer $t$ with
$0< t< d$.
For the cells containing a corner of $\meb$, we use the standard range tree on $P$
in $O(\log^{d-1} n)$ time.
In this way, we can check the emptiness for all cells intersecting
the boundary of $\meb$ in $O(\log^{d-1}n + \log n/\eps^{d-1})$ time in
total.

Now we consider the grid cells fully contained in $\meb$.  Let $Q'$ be
the union (a $d$-dimensional box) of all such grid cells, which can be
computed in constant time by a simple calculation with respect to
the coordinates of $\meb$. 
For each index $i$, consider the standard grid of side length $\eps\apx$ 
such that the union of the cells 
coincides with the projection of $Q'$ onto a hyperplane orthogonal to the $i$th axis.
Let $\mathcal{G}_i$ be the set of all such grid cells.
For each cell $\cell_i$ of $\mathcal{G}_i$, we want to
find the point $p\in P\cap Q$ with largest (and smallest) $i$th coordinate among the points
whose projections are in $\cell_i$. We choose the grid cell in the standard grid with side length $\eps\apx$
containing $p$ as a $\eps$-coreset. To do this, observe that $p$ is in $P\cap Q$
if and only if the projection of $p$ onto the $i$th axis is in $[q_i,q_i']$, where
$[q_i,q_i']$ is the projection of $Q'$ onto the $i$th axis.
Due to the data structure introduced in this section, this can be computed in $O(\log n)$ time.
Since there are
$O(1/\eps^{d-1})$ cells of $\mathcal{G}_i$ and $d$ is a constant, this
can be done in $O(\log n/\eps^{d-1})$ time in total.

Therefore, we can compute an $\eps$-coreset of $P\cap Q$ in
$O(\log^{d-1} n + \log n/\eps^{d-1})$ time in total.  The 
diameter of $N$ points can be computed
in $O(N+1/\eps^{d-1.5})$ time~\cite{Chan-diameter-2002}.
Since the size of the coreset is $O(1/\eps^{d-1})$ in our case, the overall running time is
$O(\log^{d-1} n + \log n/\eps^{d-1})$.

\subparagraph{Remark.}  An approximate radius can be computed in a
similar way. The radius of a point set $P$ is defined as
$\min_{c\in \mathbb{R}^d} \max_{p\in P} d(p,c)$. A constant-factor
approximation to the diameter of $P$ is also a constant-factor
approximation to the radius of $P$. The coreset we considered
for the diameter is also a coreset for the radius. Therefore, we can
compute an $\eps$-coreset of $P\cap Q$ for the radius in
$O(\log^{d-1} n + \log n/\eps^{d-1})$ time. Since the radius of a point set
can be computed in linear time for any fixed
dimension~\cite{Megiddo-83}, we can compute the radius of $P\cap Q$ in
$O(\log^{d-1} n + \log n/\eps^{d-1})$ time in total.

\begin{theorem}
  Given a set $P$ of $n$ points in $\mathbb{R}^d$, we can compute an
  $(1+\eps)$-approximate diameter (or radius) of $P\cap Q$ in
  $O(\log^{d-1} n + \log n/\eps^{d-1})$ time for a query consisting of
  an orthogonal range $Q$ and a value $\eps>0$ using a data structure
  of size $O(n\log^{d-1} n)$.
\end{theorem}

\newpage
\bibliography{paper}{}
\end{document}